\documentclass[journal]{IEEEtran}

\ifCLASSINFOpdf
\usepackage[pdftex]{graphicx}
\DeclareGraphicsExtensions{.pdf,.jpeg,.png}
\else
\usepackage[dvips]{graphicx}
\fi
\usepackage{epstopdf}
\usepackage{amssymb,amsmath,amsthm}
\usepackage{wasysym}
\usepackage{algorithm}
\usepackage{algorithmic}
\usepackage{array}
\usepackage{cite}
\usepackage{color}
\usepackage{url}

\usepackage{epsfig,latexsym}
\usepackage{flushend}
\usepackage{verbatim}
\usepackage{amsopn}
\usepackage{booktabs}

\usepackage{stfloats}
\usepackage{enumerate}
\usepackage{hyperref}
\usepackage{subfigure}
\usepackage{caption}
\usepackage{bm}
\usepackage{mathtools}

\newtheorem{theorem}{Theorem}

\begin{document}
	
	\title{Movable-Antenna Trajectory Optimization for Wireless Sensing: CRB Scaling Laws over Time and Space}
	
	\author{{Wenyan Ma, \IEEEmembership{Graduate Student Member, IEEE}, Lipeng Zhu, \IEEEmembership{Member, IEEE}, and  Rui Zhang, \IEEEmembership{Fellow, IEEE}}
		\thanks{W. Ma, L. Zhu, and R. Zhang are with the Department of Electrical and Computer Engineering, National University of Singapore, Singapore 117583 (e-mail: wenyan@u.nus.edu, zhulp@nus.edu.sg, elezhang@nus.edu.sg).
	}}
	\maketitle
	
	\begin{abstract}
		In this paper, we present a new  wireless sensing system utilizing a movable antenna (MA) that continuously moves and receives sensing signals to enhance sensing performance over the conventional fixed-position antenna (FPA) sensing. We show that the angle estimation performance is fundamentally determined by the MA trajectory, and derive the Cramér-Rao bound (CRB) of the mean square error (MSE) for angle-of-arrival (AoA) estimation as a function of the trajectory for both one-dimensional (1D) and two-dimensional (2D) antenna movement. For the 1D case, a globally optimal trajectory that minimizes the CRB is derived in closed form. Notably, the resulting CRB decreases cubically with sensing time in the time-constrained regime, whereas it decreases linearly with sensing time and quadratically with the movement line segment's length in the space-constrained regime. For the 2D case, we aim to achieve the minimum of maximum (min-max) CRBs of estimation MSE for the two AoAs with respect to the horizontal and vertical axes. To this end, we design an efficient alternating optimization algorithm that iteratively updates the MA’s horizontal or vertical coordinates with the other being fixed, yielding a locally optimal trajectory. Numerical results show that the proposed 1D/2D MA-based sensing schemes significantly reduce both the CRB and actual AoA estimation MSE compared to conventional FPA-based sensing with uniform linear/planar arrays (ULAs/UPAs) as well as various benchmark MA trajectories. Moreover, it is revealed that the steering vectors of our designed 1D/2D MA trajectories have low correlation in the angular domain, thereby effectively increasing the angular resolution for achieving higher AoA estimation accuracy.	
	\end{abstract}
	
	
	\section{Introduction}
	The upcoming sixth-generation (6G) mobile communication networks are anticipated to enable a wide range of location-aware services, including autonomous driving, robotic navigation, and drone scheduling \cite{jiang2021the}. These applications require advanced sensing capabilities from wireless infrastructures, which is beyond conventional quality of service (QoS) requirements for communication rate and reliability. Consequently, there is a growing interest in integrated sensing and communication (ISAC), a new paradigm in which sensing and communication functions are jointly realized using shared hardware and/or spectrum resources. It is anticipated that wireless sensing, encompassing tasks such as detection, parameter estimation, and information extraction from surrounding targets, will become a core service in future 6G networks.
	
	Achieving high angular resolution and strong beamforming gain in sensing requires large-scale antenna arrays at nodes such as radars or base stations (BSs) \cite{mailloux2005phased}. However, the associated hardware cost and power consumption increase with the number of antennas, posing challenges for implementing low-cost yet high-performance sensing systems. To mitigate these costs, sparse antenna arrays have been proposed, where fewer antennas are deployed with increased inter-element spacing to retain much of the angular resolution of a large-scale antenna array \cite{roberts2011sparse}. Nevertheless, these sparse arrays typically employ fixed-position antennas (FPAs), which lack adaptability to varying sensing requirements and cannot switch between array geometries optimized for sensing and communication tasks. Moreover, FPAs in both large-scale and sparse arrays are unable to fully exploit the spatial variations of wireless signals within the region where the sensing transmitter or receiver is located.
	
	To address the limitations of FPA-based wireless sensing, this paper proposes a new sensing system employing the movable antenna (MA) \cite{zhu2023MAMag,zhu2025tutorial}, where the antenna's position at the transmitter or receiver can be dynamically adjusted. Compared with conventional FPA arrays, this additional degree of freedom (DoF) allows for significant performance enhancement with the same or even fewer antennas. First, expanding the antenna movement region effectively increases the array aperture, improving angular resolution. Second, optimizing the MA trajectory reduces correlation between steering vectors over different directions, mitigating ambiguity and interference in angle estimation. Third, real-time position adjustments allow adaptation to changing environmental conditions and varying sensing requirements. In practice, the MA trajectory can be either pre-configured for specific sensing tasks or dynamically optimized to meet evolving sensing demands.
	
	Research on MAs in wireless communications can be traced back to 2009 \cite{zhao2009single}, where significant diversity gain was achieved by relocating the MA within a given geographical area. Recently, MAs have attracted significant research attention due to their emerging applications and demonstrated advantages in wireless communications. For example, studies in  \cite{zhu2022MAmodel,mei2024movable,ning2024movable,tang2024secure} have shown that deploying MAs within a region spanning several wavelengths can substantially improve the received signal-to-noise ratio (SNR) under both deterministic and stochastic channel models. The MA-aided multiuser communication systems have been extensively studied \cite{zhu2023MAmultiuser,wu2023movable,qin2024antenna,cheng2023sum,yang2024flexible,hu2024power}, where optimizing the positions of MAs jointly with beamforming can efficiently mitigate multiuser interference. The spatial multiplexing in MA-aided MIMO systems has been analyzed in \cite{ma2022MAmimo,chen2023joint,yeyuqi2023fluid}, while channel estimation for MA systems, which enables reconstruction of the channel response for arbitrary transmit and receive antenna locations, has been discussed in \cite{ma2023MAestimation,xiao2023channel}. Additionally, the effectiveness of MA arrays in interference nulling and multi-beam forming has been demonstrated in \cite{zhu2023MAarray,ma2024multi}, and their potential in satellite communications and near-filed communication has been investigated in \cite{ZhuLP_satellite_MA,zhu2024nearfield}. Beyond conventional MA systems, six-dimensional MA (6DMA) architectures were introduced in \cite{shao2024discrete,shao2024Mag6DMA,shao2024exploiting}, further incorporating three-dimensional (3D) antenna rotations to fully exploit the spatial DoFs offered by antenna movement.
	
	Recently, the advantages of MAs have been extended to wireless sensing. Preliminary studies have shown that adjusting the positions of MAs can enhance radar imaging quality and assist in target localization \cite{zhuravlev2015experi,hinske2008using}. However, these works primarily demonstrated performance improvements from antenna repositioning without revealing the fundamental relationship between MAs' positions and sensing performance. More recent research has focused on optimizing MAs' positions to minimize the Cramér-Rao bound (CRB) for angle-of-arrival (AoA) estimation in both far-field and near-field scenarios \cite{ma2024MAsensing,chen2025MAISACopt,wang2025MAnearsensing}, where closed-form expressions of the CRB with respect to (w.r.t.) MAs' positions are derived and then minimized. This concept has been further advanced with 6DMAs \cite{shao2024exploiting}, which incorporate 3D rotations in addition to 3D positions, fully exploiting spatial DoFs to minimize the sum CRB of targets at typical locations. Nevertheless, these studies focus on improving sensing performance by optimizing MA array geometry, which is then fixed for a given sensing task. As a result, the joint time-spatial DoFs offered by antenna movement are not fully utilized, especially when only a limited number of MAs are available.
	
	In this paper, we study the MA-aided sensing system that exploits the additional DoFs provided by optimizing the MA trajectory. Unlike prior works that focus on optimizing the MA array geometries \cite{ma2024MAsensing,chen2025MAISACopt,wang2025MAnearsensing,shao2024exploiting}, we consider the scenario where the MA continuously moves while receiving sensing signals, and we design its trajectory to synthesize a large continuous virtual array for improving sensing performance. The main contributions of this paper are summarized as follows:
	
	\begin{itemize}
		\item First, for the case of one-dimensional (1D) antenna movement, we show that the CRB of the AoA estimation mean square error (MSE) decreases inversely with the variance of all MA's positions over time. Motivated by this, we formulate an optimization problem to minimize the CRB by optimizing the MA trajectory. Although the formulated problem is non-convex, we derive a globally optimal trajectory solution in closed form. Notably, the resulting CRB decreases cubically with sensing time in the time-constrained (TC) regime, whereas it decreases linearly with sensing time and quadratically with the movement line segment's length in the space-constrained (SC) regime.
		\item Next, for the case of two-dimensional (2D) antenna movement, we derive the CRB of AoA estimation MSE and formulate a trajectory optimization problem to minimize the CRB. Specifically, the objective is to achieve the minimum of maximum (min-max) CRBs of estimation MSE for the two AoAs w.r.t. the horizontal and vertical axes, respectively. To this end, we design an efficient alternating optimization algorithm that iteratively updates the MA’s horizontal or vertical coordinates with the other being fixed, yielding a locally optimal trajectory.
		\item Finally, we provide extensive simulation results showing that the proposed MA-aided sensing scheme achieves lower AoA estimation MSE compared with conventional FPA-based sensing with uniform linear/planar arrays (ULAs/UPAs) as well as various benchmark MA trajectories. The results further reveal that the designed 1D/2D MA trajectories yield steering vectors with low angular-domain correlation, thereby mitigating estimation ambiguity and enhancing AoA estimation accuracy.
	\end{itemize}
	
	The rest of this paper is organized as follows. Section II and Section III introduce the system model and analyze the CRB of AoA estimation MSE for 1D and 2D antenna movement, respectively. In Section IV, an alternating optimization algorithm is proposed to address the trajectory optimization problem for the 2D case. Section V presents numerical results as well as detailed discussions, and Section VI concludes the paper.
	
	\textit{Notations}: Vectors (lowercase) and matrices (uppercase) are denoted in boldface. The transpose and Hermitian transpose are represented by $(\cdot)^{\mathsf T}$ and $(\cdot)^{\mathsf H}$, respectively. The sets of $P \times Q$ complex and real matrices are denoted by $\mathbb{C}^{P \times Q}$ and $\mathbb{R}^{P \times Q}$, respectively. $\bm{I}_N$ represents the $N$-dimensional identity matrix. For any real number $a$, $\lceil a \rceil$, $\lfloor a \rfloor$, and $\langle a \rangle$ denote its ceiling, floor, and rounded value, respectively. The $\ell_2$-norm of vector $\bm{a}$ is given by $\|\bm{a}\|_2$.

	\begin{figure}[!t]
		\centering
		\includegraphics[width=75mm]{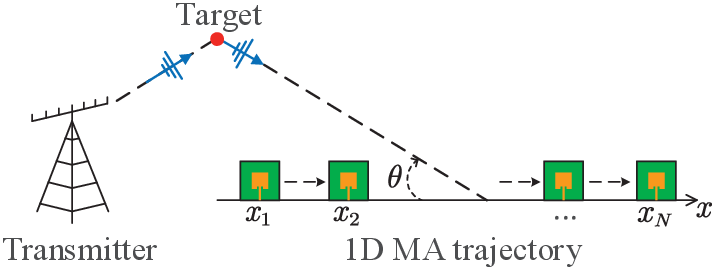}
		\caption{System model for the 1D MA-aided sensing system.}
		\label{fig_model_1D}
	\end{figure}
	
	\section{CRB Characterization for 1D Antenna Movement}
	
	\subsection{System Model}
		As shown in Fig.~\ref{fig_model_1D}, we consider a bistatic sensing system where a receiver is equipped with a single MA constrained to move along a 1D line segment of length $A$. To estimate the AoA of the target, the sensing transmitter continuously emits probing signals, while the single MA continuously moves and receives the echoes reflected by the target over $N$ snapshots. The MA receives one signal per snapshot, with the signal interval denoted by $T_s$. Thus, the total time for antenna movement and sensing is denoted by $T=NT_s$. We set the start time $t_1 = 0$ when the MA is at its initial position $x_1\in[0,A]$\footnote{Due to the limited moving speeds of practical MAs and the short duration of each probing signal, it is justified to assume that the MA's position remains unchanged within each snapshot.}, while the time of receiving the $n$th ($n=1, \dots, N$) signal is denoted by $t_n = (n-1)T_s$. The position of the MA when receiving the $n$th signal is denoted by $x_n\in[0,A]$, and the antenna position vector (APV) over the $N$ snapshots is denoted by $\bm{x} \coloneqq [x_1, \dots, x_N]^{\mathsf T}$. The antenna's velocity between position $x_n$ and $x_{n+1}$ is denoted by $v_n$ ($n=1, \dots, N-1$), which is constrained by a maximum speed $v^{\rm m}$ with $v^{\rm m}>0$, i.e., $-v^{\rm m}<v_n<v^{\rm m}$. Then, the antenna position $x_n$ ($n=2, \dots, N$) can be expressed as a function of its initial position $x_1$ and the velocity sequence $\{v_n\}_{n=1}^{N-1}$:
	\begin{align}
		x_n &= x_{n-1} + v_{n-1}(t_n - t_{n-1}) \notag\\
		&= x_1 + T_s \sum_{m=1}^{n-1} v_m.
	\end{align}
	We consider that the target-receiver link is dominated by a line-of-sight
	(LoS) channel path, which is invariant over $N$ snapshots. Since the distance between the target and receiver is much
	larger than the region size for antenna movement, the far-field channel model is considered for the target-receiver link. As illustrated in Fig.~\ref{fig_model_1D}, the physical AoA of
	the target is denoted by $\theta \in [0, \pi]$. For
	simplicity, we define the spatial AoA as $u \coloneqq \cos\theta$. Then, the steering vector of the MA over $N$ snapshots can be represented as a function of the APV $\bm{x}$ and the spatial AoA $u$, i.e.,
	\begin{equation}
		\bm{a}(\bm{x}, u) \coloneqq \left[ e^{j\frac{2\pi}{\lambda}x_1 u}, \dots, e^{j\frac{2\pi}{\lambda}x_N u} \right]^{\mathsf T}\in \mathbb{C}^{N\times1},
	\end{equation}
	where $\lambda$ is the carrier wavelength. Moreover, Let $\beta$ denote the complex channel coefficient of the target–receiver link measured at the origin of the line segment. Accordingly, the LoS channel vector between the target and receiver can be expressed as
	\begin{equation}
		\bm{h} = \beta \bm{a}(\bm{x}, u).	
	\end{equation}
	We assume that the transmitter continuously emits a fixed omnidirectional probing signal. Then, the vector of $N$ received signals can be written as
	\begin{equation}
		\bm{y} = \bm{h}s + \bm{z},
	\end{equation}
	where $s$ is the reflected signal from the target with average power $P = \mathbb{E}\{|s|^2\}$, and $\bm{z} \sim \mathcal{CN}(0, \sigma^2 \bm{I}_N)$ is the additive white Gaussian noise (AWGN) vector at the receiver.

	\subsection{AoA Estimation}
	For a given MA trajectory $\bm{x}$, we use the maximum likelihood estimation (MLE) method to estimate the spatial AoA $u$. Specifically, the unknown parameters $\tilde{\beta}\coloneqq\beta s$ and $u$ are jointly estimated by solving
	\begin{align}\label{hatbetaeta}
		\left(\hat{\beta},\hat{u}\right) = \arg\min_{\bar{\beta},\bar{u}} \|\bm{y} - \bar{\beta}\bm{a}(\bm{x},\bar{u})\|_2^2.
	\end{align}
	For any fixed $u$, the optimal estimation of $\tilde{\beta}$ is given by
	\begin{align}\label{hatbeta}
		\hat{\beta} = \frac{\bm{a}(\bm{x},u)^{\mathsf H} \bm{y}}{\|\bm{a}(\bm{x},u)\|_2^2}.
	\end{align}
	Then, substituting \eqref{hatbeta} into \eqref{hatbetaeta} yields
	\begin{align}
		&\|\bm{y} - \hat{\beta}\bm{a}(\bm{x},u)\|_2^2 \\
		=& \left(\bm{y} - \hat{\beta}\bm{a}(\bm{x},u)\right)^{\mathsf H} \left(\bm{y} - \hat{\beta}\bm{a}(\bm{x},u)\right) \notag\\
		=&\|\bm{y}\|_2^2 + |\hat{\beta}|^2 \|\bm{a}(\bm{x},u)\|_2^2  - 2\Re\left\{\hat{\beta}\bm{y}^{\mathsf H} \bm{a}(\bm{x},u)\right\} \notag\\
		=&\|\bm{y}\|_2^2 - \frac{\left|\bm{y}^{\mathsf H} \bm{a}(\bm{x},u)\right|^2}{\|\bm{a}(\bm{x},u)\|_2^2} \notag\\
		=&\|\bm{y}\|_2^2 - \frac{1}{N}\left|\bm{y}^{\mathsf H} \bm{a}(\bm{x},u)\right|^2. \notag
	\end{align}
	Since $\|\bm{y}\|_2^2$ is independent of $u$, the MLE of $u$ is equivalently expressed as
	\begin{align}\label{MLE1D}
		\hat{u} &= \arg\min_{\bar{u}\in[-1,1]} \|\bm{y}\|_2^2 - \frac{1}{N}\left|\bm{y}^{\mathsf H} \bm{a}(\bm{x},u)\right|^2 \\
		&= \arg\max_{\bar{u}\in[-1,1]} \left|\bm{y}^{\mathsf H} \bm{a}(\bm{x},u)\right|^2, \notag
	\end{align}
	which can be solved by exhaustively searching for $\bar{u}$ over the interval $[-1,1]$. Then, the MSE of AoA estimation is given by	
	\begin{align}
		{\rm MSE}(u) \coloneqq \mathbb{E}{|u-\hat{u}|^2}.
	\end{align}	
	Thus, the CRB that serves as the lower-bound of ${\rm MSE}(u)$ can be expressed as \cite{ma2024MAsensing,ma2025MAISAC}	
	\begin{align}\label{CRB1D}
		{\rm MSE}(u) \geq {\rm CRB}_u(\bm{x}) = \frac{\sigma^2\lambda^2}{8\pi^2PN|\beta|^2} \frac{1}{{\rm var}(\bm{x})},
	\end{align}	
	where the variance function is defined as
	${\rm var}(\bm{x}) \coloneqq \frac{1}{N}\sum_{n=1}^{N}x_n^2 - \mu(\bm{x})^2$,
	with $\mu(\bm{x})=\frac{1}{N}\sum_{n=1}^{N}x_n$ denoting the mean of the APV $\bm{x}$. The result in \eqref{CRB1D} shows that the CRB scales inversely with the variance of the antenna positions, which characterizes the effective aperture of the time-synthesized virtual array. Consequently, minimizing the CRB requires maximizing this variance through trajectory design, determined by the initial position $x_1$ and the antenna velocity vector (AVV) $\bm{v} \coloneqq [v_1, \dots, v_{N-1}]^{\mathsf T}$.
	
	\subsection{CRB Minimization}
	In this subsection, our goal is to minimize ${\rm CRB}_u(\bm{x})$ by optimizing the APV $\bm{x}$ and AVV $\bm{v}$. From \eqref{CRB1D}, this objective is equivalently expressed as
	\begin{align}
		\min_{\bm{x},\bm{v}}{\rm{CRB}}_u(\bm{x}) \iff \max_{\bm{x},\bm{v}}{\rm{var}}(\bm{x}).
	\end{align}	
	The optimization problem is thus formulated as
	\begin{subequations}
		\begin{align}
			\textrm{(P1)}~~\max_{\bm{x},\bm{v}} \quad & \text{var}(\bm{x}) \label{P1a}\\
			\text{s.t.} \quad & -v^{\rm m} \le v_n \le v^{\rm m},~~ n=1, \dots, N-1, \label{P1b}\\
			& 0 \le x_n \le A,~~ n=1, \dots, N, \label{P1c}\\
			& x_n = x_1 + T_s \sum_{m=1}^{n-1} v_m,~~ n=2, \dots, N. \label{P1d}
		\end{align}
	\end{subequations}
	Problem (P1) is non-convex due to the non-concave objective function. Let $\Delta \coloneqq v^{\rm m}T_s$ denote the maximum spacing of MA's positions over two adjacent snapshots. The globally optimal solution to problem (P1) is given by the following theorems for two distinct cases, respectively.
	\begin{theorem}
		In the \textbf{\textit{time-constrained (TC)}} regime with $A \geq (N-1)\Delta$, an optimal AVV $\bm{v}^{\rm TC}$ that minimizes the CRB of AoA estimation MSE along the 1D line segment is
		\begin{equation}\label{vTC}
			v_n^{\rm TC} = v^{\rm m}, \quad n=1,\ldots,N-1.
		\end{equation}
		The corresponding APV $\bm{x}^{\rm TC}$ is
		\begin{align}\label{xTC}
			x_n^{\rm TC} = x_1^{\rm TC} + (n-1)\Delta, \quad n=2,\ldots,N,
		\end{align}
		where $x_1^{\rm TC}$ satisfies
		\begin{align}
			0 \leq x_1^{\rm TC} \leq A - (N-1)\Delta.
		\end{align}
		The resulted minimum CRB is given by
		\begin{align}\label{crbinf_main}
			{\rm CRB}_u\left(\bm{x}^{\rm TC}\right) = \frac{3\sigma^2\lambda^2}{2\pi^2P|\beta|^2\Delta^2}  \frac{1}{N(N^2-1)},
		\end{align}		
		which decreases with $N$ in the order of $\mathcal{O}(N^{-3})$.
	\end{theorem}
	\begin{proof}
		See Appendix A.
	\end{proof}

	\begin{theorem}		
		In the \textbf{\textit{space-constrained (SC)}} regime with $A < (N-1)\Delta$, the optimal AVV $\bm{v}^{\rm SC}$ is
		\begin{align}\label{vRL_main}
			v_n^{\rm SC} =
			\begin{cases}
				0, & n=1,\ldots,N_{\rm L}-1, \\[0.5em]
				v^{\rm m}, & n=N_{\rm L},\ldots,N-N_{\rm R}-1, \\[0.5em]
				\dfrac{A-N_{\rm M}\Delta}{T_s}, & n=N-N_{\rm R}, \\[0.5em]
				0, & n=N-N_{\rm R}+1,\ldots,N-1,
			\end{cases}
		\end{align}
		with $N_{\rm M} = \left\lceil\frac{A}{\Delta}\right\rceil - 1$,  
		$N_{\rm L} = \left\lceil \frac{N-N_{\rm M}}{2} \right\rceil$,  
		and $N_{\rm R} = \left\lfloor \frac{N-N_{\rm M}}{2} \right\rfloor$. The corresponding APV $\bm{x}^{\rm SC}$ is
		\begin{align}\label{xRL3_main}
			x_n^{\rm SC} =
			\begin{cases}
				0, & n=1,\ldots,N_{\rm L}, \\[0.5em]
				(n-N_{\rm L})\Delta, & n=N_{\rm L}+1,\ldots,N-N_{\rm R}, \\[0.5em]
				A, & n=N-N_{\rm R}+1,\ldots,N.
			\end{cases}
		\end{align}
		The resulted minimum CRB, given in \eqref{crbfin1} and \eqref{crbfin2} of Appendix B, decreases with $N$ in the order of $\mathcal{O}(N^{-1})$ and decreases with $A$ in the order of $\mathcal{O}(A^{-2})$.
	\end{theorem}
	\begin{proof}
		See Appendix B.
	\end{proof}	
	For each case, there exists a symmetric counterpart obtained by mirroring the APV and reversing the velocity directions. Specifically, for $(\bm{x},\bm{v}) \in \{(\bm{x}^{\rm TC},\bm{v}^{\rm TC}), (\bm{x}^{\rm SC},\bm{v}^{\rm SC})\}$, the mirrored optimal solution $(\tilde{\bm{x}},\tilde{\bm{v}})$ is given by
	\begin{align}
		\tilde{x}_n &= A - x_{N+1-n}, \quad n=1,\ldots,N, \notag\\
		\tilde{v}_n &= -v_{N-n}, \quad n=1,\ldots,N-1.
	\end{align}

	Theorem~1 and Theorem~2 reveal different structures in the optimal trajectory design depending on whether the line segment length $A$ is sufficiently large relative to the maximum trajectory distance of the MA. In the \textbf{\textit{TC}} regime with $A \geq (N-1)\Delta$, the optimal strategy is to let the MA move continuously with the maximum speed $v^{\rm m}$, thereby generating a uniformly spaced APV that maximizes the virtual array aperture as shown in Fig.~\ref{optimal_position_1D}(a). This configuration achieves a CRB scaling order of $\mathcal{O}(N^{-3})$, where the $\mathcal{O}(N^{-2})$ component results from the variance of the virtual aperture created by antenna movement, while the additional $\mathcal{O}(N^{-1})$ component results from the accumulation of sensing signals over time. In contrast, in the \textbf{\textit{SC}} regime with $A < (N-1)\Delta$, the optimal trajectory partitions the APV into three groups as shown in Fig.~\ref{optimal_position_1D}(b): a subset fixed at the left-most position, a subset fixed at the right-most position, and a middle subset moving with maximum speed. This hybrid arrangement maximizes the variance of MA's positions within the limited aperture, leading to a CRB scaling order of $\mathcal{O}(N^{-1}A^{-2})$, where the $\mathcal{O}(N^{-1})$ component results from the sensing time, while the $\mathcal{O}(A^{-2})$ component results from the variance of the maximum aperture $A$. Overall, Theorem~1 and Theorem~2 reveal the importance of increasing both the sensing time and antenna movement region size for decreasing the CRB.

	\begin{figure}[!t]
		\centering
		\includegraphics[width=70mm]{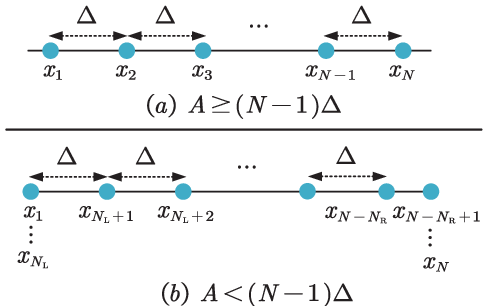}
		\caption{Illustration of the optimal MA trajectories for 1D antenna movement.}
		\label{optimal_position_1D}
	\end{figure}

	\subsection{Comparison with FPA Array}	
	In this subsection, we compare the proposed single MA-aided sensing system with a conventional FPA array-aided sensing system. Consider an FPA receiver equipped with an $M$-antenna ULA with half-wavelength spacing, where the APV of $M$ antennas is denoted by $\bm{x}_{\rm FPA}\in\mathbb{R}^{M\times1}$. The FPA receives $N$ reflected signals from the target to estimate the spatial AoA $u$, with the corresponding CRB given by \cite{ma2024MAsensing,ma2025MAISAC}
	\begin{align}\label{crbFPA1D}
		\overline{\rm CRB}_u(\bm{x}_{\rm FPA}) &=\frac{\sigma^2 \lambda^2}{8 \pi^2 P |\beta|^2MN} \frac{1}{\text{var}(\bm{x}_{\rm FPA})} \notag\\		
		&=\frac{\sigma^2 \lambda^2}{8 \pi^2 P |\beta|^2MN} \frac{12}{\left(\frac{\lambda}{2}\right)^2 (M^2-1)} \notag\\
		&= \frac{6 \sigma^2}{\pi^2 P |\beta|^2} \frac{1}{N M(M^2-1)}.
	\end{align}
	This indicates that an $M$-antenna 1D FPA array achieves a CRB scaling order of $\mathcal{O}(N^{-1}M^{-3})$, where the $\mathcal{O}(M^{-2})$ component results from the array aperture, the additional $\mathcal{O}(M^{-1})$ component from the array gain, and the $\mathcal{O}(N^{-1})$ component from the sensing time. Considering the case of $A \geq (N-1)\Delta$ (i.e., the TC regime), and using the approximations $N^2-1 \approx N^2$ and $M^2-1 \approx M^2$ for sufficiently large $N$ and $M$, the relative performance of the single MA-aided sensing system compared to the FPA array-aided system can be expressed as
	\begin{align}
		\frac{{\rm CRB}_u\left(\bm{x}^{\rm TC}\right)}{\overline{\rm CRB}_u(\bm{x}_{\rm FPA})} \approx \frac{\lambda^2 M^3 }{4 \Delta^2 N^2}.
	\end{align}
	Thus, for the MA system to outperform the conventional FPA system, i.e., ${\rm CRB}_u\left(\bm{x}^{\rm TC}\right) < \overline{\rm CRB}_u(\bm{x}_{\rm FPA})$, the number of snapshots $N$ should satisfy
	\begin{align}
		N > \frac{M^{\frac{3}{2}}}{2} \frac{\lambda}{\Delta}.
	\end{align}
	Since $T = NT_s$ and $\Delta = v^{\rm m}T_s$, this also implies
	\begin{align}
		T > \frac{M^{\frac{3}{2}}}{2} \frac{\lambda}{v^{\rm m}}.
	\end{align}
	This comparison demonstrates that, with sufficient antenna movement speed $v^{\rm m}$ and sensing time $T$ (as well as antenna movement line length), the single MA-aided system can achieve superior AoA estimation accuracy compared to a conventional FPA array. For instance, with $v^{\rm m} = 1$ m/s and $\lambda = 0.01$ m, the single MA system can outperform the conventional FPA array with $M=64$ antennas when $T>2.56$ s, which is easily attainable in practical scenarios with low-mobility targets.

	\begin{figure}[!t]
		\centering
		\includegraphics[width=75mm]{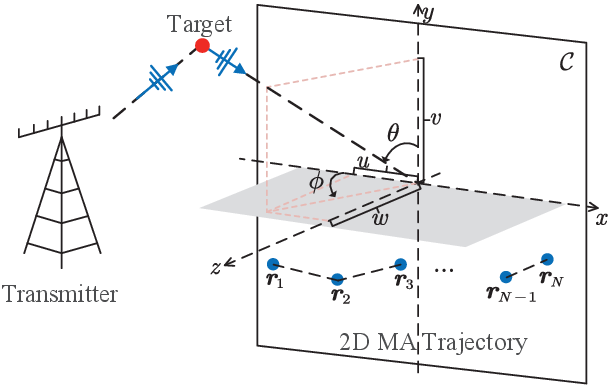}
		\caption{System model for the 2D MA-aided sensing system.}
		\label{fig_model_2D}
	\end{figure}
	
	\section{CRB Characterization for 2D Antenna Movement}
	
	\subsection{System Model}
	As illustrated in Fig.~\ref{fig_model_2D}, we consider a wireless sensing system with an MA moving on a 2D plane to estimate the target’s AoAs w.r.t. the $x$ and $y$ axes. The MA moves continuously within a 2D region denoted by $\mathcal{C}$. Let $\bm{r}_n = [x_n, y_n]^{\mathsf T} \in \mathcal{C}$ denote the MA's coordinates at the $n$th snapshot ($n=1,2,\ldots,N$), and the antenna position matrix (APM) over $N$ snapshots is denoted by $\bm{R} = [\bm{r}_1, \bm{r}_2, \ldots, \bm{r}_N] \in \mathbb{R}^{2\times N}$. The trajectory is determined by the initial position $\bm{r}_1$ and the antenna velocity matrix (AVM) $\bm{V} = [\bm{v}_1, \dots, \bm{v}_{N-1}] \in \mathbb{R}^{2\times (N-1)}$, where $\bm{v}_n = [v_n^x, v_n^y]^{\mathsf T}$ and $\|\bm{v}_n\|_2 \le v^{\rm m}$. The MA's position $\bm{r}_n$ ($n=2,\ldots,N$) is thus updated according to
	\begin{equation}
		\bm{r}_n = \bm{r}_1 + T_s \sum_{m=1}^{n-1} \bm{v}_m.
	\end{equation}
	The physical elevation and azimuth AoAs of the target–receiver LoS path are denoted by $\theta \in [0, \pi]$ and $\phi \in [0, \pi]$, respectively. For convenience, we define the corresponding spatial AoAs w.r.t. the $x$ and $y$ axes as
	\begin{equation}
		u = \sin\theta \cos\phi, \quad v = \cos\theta.
	\end{equation}
	Accordingly, the steering vector associated with the 2D MA trajectory can be expressed as a function of the MA's positions $\bm{R}$ and the spatial AoAs $\bm{\eta} \coloneqq [u, v]^{\mathsf T}$:
	\begin{equation}\label{g2}
		\bm{\alpha}(\bm{R},\bm{\eta}) \coloneqq \left[ e^{j\frac{2\pi}{\lambda}(x_1 u + y_1 v)}, \ldots, e^{j\frac{2\pi}{\lambda}(x_N u + y_N v)} \right]^{\mathsf T} \in \mathbb{C}^{N\times1}.
	\end{equation}
	The LoS channel between the target and receiver is thus given by
	\begin{equation}\label{H2}
		\bm{h}(\bm{R},\bm{\eta}) = \beta \bm{\alpha}(\bm{R}, \bm{\eta}).
	\end{equation}

	\subsection{AoA Estimation and CRB Minimization}
	Following the approach for 1D antenna movement, the joint MLE of the spatial AoAs $u$ and $v$ for 2D antenna movement is given by \cite{shao2022target}
	\begin{equation}\label{MLE2D}
		\hat{\bm{\eta}} = \arg\max_{\bar{\bm{\eta}} \in [-1,1]\times[-1,1]} \left| \bm{y}^{\mathsf H} \bm{\alpha}(\bm{R}, \bar{\bm{\eta}}) \right|^2,
	\end{equation}
	which can be implemented via exhaustive search over $\bar{\bm{\eta}} = [\bar{u}, \bar{v}]^{\mathsf T}$ within the interval $[-1,1]\times[-1,1]$.
	For the MA-aided sensing system with 2D antenna movement region $\mathcal{C}$, the two CRBs for AoA estimation w.r.t. the $x$ and $y$ axes are given by \cite{ma2024MAsensing,ma2025MAISAC}
	\begin{align}\label{CRBr}
		{\rm MSE}(u) \ge {\rm CRB}_u(\bm{R}) &= \frac{\sigma^2 \lambda^2}{8 \pi^2 P  |\beta|^2N} \frac{1}{{\rm var}(\bm{x}) - \frac{{\rm cov}(\bm{x}, \bm{y})^2}{{\rm var}(\bm{y})}}, \notag\\
		{\rm MSE}(v) \ge {\rm CRB}_v(\bm{R}) &= \frac{\sigma^2 \lambda^2}{8 \pi^2 P  |\beta|^2N} \frac{1}{{\rm var}(\bm{y}) - \frac{{\rm cov}(\bm{x}, \bm{y})^2}{{\rm var}(\bm{x})}},
	\end{align}
	where $\bm{x} = [x_1, \dots, x_N]^{\mathsf T}$, $\bm{y} = [y_1, \dots, y_N]^{\mathsf T}$, and the covariance function is defined as ${\rm cov}(\bm{x},\bm{y}) = \frac{1}{N} \sum_{n=1}^N x_n y_n - \mu(\bm{x}) \mu(\bm{y})$.
	Equation \eqref{CRBr} indicates that the two CRBs depend explicitly on the MA trajectory $\bm{R}$. In particular, the two CRBs decrease as ${\rm var}(\bm{x}) - \frac{{\rm cov}(\bm{x},\bm{y})^2}{{\rm var}(\bm{y})}$ and ${\rm var}(\bm{y}) - \frac{{\rm cov}(\bm{x},\bm{y})^2}{{\rm var}(\bm{x})}$ increase, respectively. Therefore, the MA trajectory can be optimized to jointly minimize ${\rm CRB}_u(\bm{R})$ and ${\rm CRB}_v(\bm{R})$. Intuitively, this can be achieved by maximizing the spreads of $\bm{x}$ and $\bm{y}$ while ensuring symmetry of the trajectory w.r.t. the $x$ and $y$ axes to reduce ${\rm cov}(\bm{x}, \bm{y})$. Nevertheless, a trade-off generally exists between minimizing ${\rm CRB}_u(\bm{R})$ and ${\rm CRB}_v(\bm{R})$ due to the coupling between the variances and the covariance of the trajectory components.

	Our goal is to minimize the maximum CRB of the estimation MSE for $u$ and $v$ (i.e., achieve the min-max CRB) by optimizing the APM $\bm{R}$ and AVM $\bm{V}$. From \eqref{CRBr}, this objective can be equivalently expressed as
	\begin{align}
		&\min_{\bm{R},\bm{V}}\max  ~ [{\rm{CRB}}_u(\bm{R}), {\rm{CRB}}_v(\bm{R})] \iff \\
		& \max_{\bm{R},\bm{V}}\min  ~ \left[{\rm{var}}(\bm{x})-\frac{{\rm{cov}}(\bm{x},\bm{y})^2}{{\rm{var}}(\bm{y})}, {\rm{var}}(\bm{y})-\frac{{\rm{cov}}(\bm{x},\bm{y})^2}{{\rm{var}}(\bm{x})}\right].\notag
	\end{align}
	Accordingly, the 2D antenna trajectory optimization problem can be formulated as follows:
	\begin{subequations}
		\begin{align}
			\textrm {(P2)}~~\max_{\bm{R},\bm{V},\delta} \quad & \delta \label{P2a}\\
			\text{s.t.} \quad & {\rm{var}}(\bm{x})-\frac{{\rm{cov}}(\bm{x},\bm{y})^2}{{\rm{var}}(\bm{y})} \geq \delta, \label{P2b}\\
			&{\rm{var}}(\bm{y})-\frac{{\rm{cov}}(\bm{x},\bm{y})^2}{{\rm{var}}(\bm{x})} \geq \delta, \label{P2c}\\
			& \|\bm{v}_n\|_2 \le v^{\rm m}, \quad n=1, \dots, N-1, \label{P2d} \\
			& \bm{r}_n \in \mathcal{C}, \quad n=1, \dots, N, \label{P2e} \\
			& \bm{r}_n = \bm{r}_1 + T_s \sum_{m=1}^{n-1} \bm{v}_m, \quad n=2, \dots, N. \label{P2f}
		\end{align}
	\end{subequations}
	Since the fractional constraints in \eqref{P2b} and \eqref{P2c} are non-convex w.r.t. $\bm{x}$ and $\bm{y}$, problem (P2) is a non-convex problem and it is difficult to obtain the optimal solution for it. Furthermore, the inherent coupling between $\bm{x}$ and $\bm{y}$ increases the complexity of solving problem (P2).
	
	\subsection{Comparison with FPA Array}
	In this subsection, we compare the proposed single MA-aided sensing system over a 2D antenna movement region with a conventional 2D FPA array-aided sensing system. Since problem (P2) is difficult to be solved optimally, we adopt a heuristic design where the MA moves at maximum speed $v^{\rm m}$ along a circular trajectory with radius
	\begin{align}
		R^{\rm cir} = \frac{\Delta}{2\sin\left(\frac{\pi}{N}\right)}.
	\end{align}
	Accordingly, the largest inscribed circle of the antenna movement region $\mathcal{C}$ must be no smaller than $R^{\rm cir}$. The MA’s position is then given by
	\begin{align}
		x_n &= R^{\rm cir}\cos \left(\tfrac{2\pi n}{N}\right), \quad n=1,\ldots,N, \notag\\
		y_n &= R^{\rm cir}\sin \left(\tfrac{2\pi n}{N}\right), \quad n=1,\ldots,N.
	\end{align}
	This yields
	\begin{align}
		&{\rm{CRB}}_u(\bm{R}) = {\rm{CRB}}_v(\bm{R}) = \frac{\sigma^2 \lambda^2}{\pi^2 P |\beta|^2\Delta^2} \frac{\sin\left(\frac{\pi}{N}\right)^2}{N} \notag\\
		&\overset{(a)}\approx \frac{\sigma^2 \lambda^2}{P |\beta|^2\Delta^2} \frac{1}{N^3},
	\end{align}
	where $(a)$ follows from $\sin(\pi/N)\approx\pi/N$ for sufficiently large $N$. This shows that the circular trajectory achieves a CRB scaling order of $\mathcal{O}(N^{-3})$, where the $\mathcal{O}(N^{-2})$ component results from the variance of virtual aperture and the additional $\mathcal{O}(N^{-1})$ component from sensing time.
	
	Now consider an FPA receiver equipped with a $\sqrt{M}\times\sqrt{M}$ UPA with half-wavelength spacing, where the APM denoted by $\bm{R}_{\rm FPA}=\begin{bmatrix}
		\bm{x}_{\rm FPA}^{\mathsf T} \\
	\bm{y}_{\rm FPA}^{\mathsf T}
	\end{bmatrix}\in\mathbb{R}^{2\times M}$ . The FPA receives $N$ reflected signals to estimate the spatial AoAs $u$ and $v$, with CRB given by \cite{ma2024MAsensing,ma2025MAISAC}
	\begin{align}
		&\overline{\rm CRB}_u(\bm{R}_{\rm FPA}) = \overline{\rm CRB}_v(\bm{R}_{\rm FPA}) =\frac{\sigma^2 \lambda^2}{8 \pi^2 P |\beta|^2MN} \frac{1}{\text{var}(\bm{x}_{\rm FPA})} \notag\\
		&=\frac{\sigma^2 \lambda^2}{8 \pi^2 P |\beta|^2MN} \frac{12}{\left(\frac{\lambda}{2}\right)^2 (M-1)} \notag\\
		&= \frac{6 \sigma^2}{\pi^2 P |\beta|^2} \frac{1}{N M(M-1)}.
	\end{align}
	Thus, a $\sqrt{M}\times\sqrt{M}$ UPA achieves a CRB scaling order of $\mathcal{O}(N^{-1}M^{-2})$, where the $\mathcal{O}(M^{-2})$ component results from both the variance of aperture and array gain, and the $\mathcal{O}(N^{-1})$ component from sensing time.
	Using the approximations $\sin(\pi/N)\approx \pi/N$ and $M-1\approx M$ for large $N$ and $M$, the relative performance of the single MA-aided system (assuming the circular trajectory) over the FPA array-aided system can be expressed as
	\begin{align}
		\frac{{\rm CRB}_u(\bm{R})}{\overline{\rm CRB}_u(\bm{R}_{\rm FPA})} = \frac{{\rm CRB}_v(\bm{R})}{\overline{\rm CRB}_v(\bm{R}_{\rm FPA})}
		\approx \frac{\pi^2 \lambda^2 M^2}{6 \Delta^2 N^2}.
	\end{align}
	Hence, for the considered MA system with a circular trajectory to outperform the conventional FPA system, i.e., ${\rm CRB}_u(\bm{R}) < \overline{\rm CRB}_u(\bm{R}_{\rm FPA})$ and ${\rm CRB}_v(\bm{R}) < \overline{\rm CRB}_v(\bm{R}_{\rm FPA})$, the number of snapshots $N$ should satisfy
	\begin{align}
		N > \frac{\pi M}{\sqrt{6}} \frac{\lambda}{\Delta}.
	\end{align}
	Since $T=NT_s$ and $\Delta=v^{\rm m}T_s$, this condition equivalently requires
	\begin{align}
		T > \frac{\pi M}{\sqrt{6}} \frac{\lambda}{v^{\rm m}}.
	\end{align}
	This comparison shows that, with sufficient antenna movement speed $v^{\rm m}$ and sensing time $T$ (as well as antenna movement region size), the single MA-aided system can outperform a conventional FPA array in AoA estimation accuracy. For example, with $v^{\rm m}=1$ m/s and $\lambda=0.01$ m, the MA system outperforms an FPA array with $M=64$ antennas when $T>0.82$ s, which is smaller and thus more easily attainable in practice compared to that of the previous 1D movement case.
		
	\section{Trajectory Optimization for 2D Antenna Movement}
	Although the constraints \eqref{P2d}, \eqref{P2e}, and \eqref{P2f} in problem (P2) are convex w.r.t. $\bm{V}$, optimizing the MA trajectory in problem (P2) remains challenging due to the non-convex constraints \eqref{P2b} and \eqref{P2c} as well as the coupling between $\bm{x}$ and $\bm{y}$. To tackle this problem, we propose an alternating optimization algorithm that iteratively updates either one of the horizontal APV $\bm{x}$ and the vertical APV $\bm{y}$ while keeping the other fixed, with the convex constraints \eqref{P2d}, \eqref{P2e}, and \eqref{P2f} w.r.t. $\bm{V}$ preserved in each subproblem. We assume that $\mathcal{C}$ is a convex 2D region, which ensures that the constraint in \eqref{P2e} remains convex. In addition, the fractional constraints \eqref{P2b} and \eqref{P2c} are reformulated into quadratic expressions. Specifically, ${\rm{var}}(\bm{x})$, ${\rm{var}}(\bm{y})$, and ${\rm{cov}}(\bm{x},\bm{y})$ can be equivalently expressed as
	\begin{align}
		{\rm{var}}(\bm{x}) &\coloneqq \bm{x}^{\mathsf T} \bm{B} \bm{x}, \notag\\
		{\rm{var}}(\bm{y}) &\coloneqq \bm{y}^{\mathsf T} \bm{B} \bm{y}, \notag\\
		{\rm{cov}}(\bm{x},\bm{y}) &\coloneqq \bm{x}^{\mathsf T} \bm{B} \bm{y},
	\end{align}
	where $\bm{B}\coloneqq \frac{1}{N}\bm{I}_N-\frac{1}{N^2}\bm{1}_N\bm{1}_N^{\mathsf T}$ is a positive semi-definite (PSD) matrix. Thus, the constraints in \eqref{P2b} and \eqref{P2c} can be equivalently rewritten as
	\begin{subequations}
		\begin{align}
			& G(\bm{x},\bm{y})\coloneqq G_1(\bm{x})-  G_2(\bm{x}, \bm{y})\coloneqq\bm{x}^{\mathsf T}\bm{B}\bm{x}-\frac{\left(\bm{x}^{\mathsf T}\bm{B}\bm{y}\right)^2}{\bm{y}^{\mathsf T}\bm{B}\bm{y}} \geq \delta, \label{xBx}\\
			&G(\bm{y},\bm{x}) = G_1(\bm{y})-  G_2(\bm{y}, \bm{x}) = \bm{y}^{\mathsf T}\bm{B}\bm{y}-\frac{\left(\bm{y}^{\mathsf T}\bm{B}\bm{x}\right)^2}{\bm{x}^{\mathsf T}\bm{B}\bm{x}} \geq \delta, \label{yBy}
		\end{align}
	\end{subequations}
	where $G_1(\bm{x})\coloneqq\bm{x}^{\mathsf T}\bm{B}\bm{x}$ and $G_2(\bm{x}, \bm{y})\coloneqq \frac{\left(\bm{x}^{\mathsf T}\bm{B}\bm{y}\right)^2}{\bm{y}^{\mathsf T}\bm{B}\bm{y}}$.
	
	\subsection{Optimization of $\bm{x}$ with Given $\bm{y}$}
	In this subsection, we optimize $\bm{x}$ and $\bm{V}$ in problem (P2) with fixed $\bm{y}$. Since the constraints in \eqref{xBx} and \eqref{yBy} are non-convex w.r.t. $\bm{x}$, we employ the successive convex approximation (SCA) method to convert them into convex constraints. Specifically, let $\bm{x}^p \in \mathbb{R}^{N\times1}$ denote the solution obtained at the $p$th iteration of SCA. Since $G_1(\bm{x})$ is convex w.r.t. $\bm{x}$, it can be globally lower-bounded using its first-order Taylor expansion at $\bm{x}^p$, i.e.,
	\begin{align}\label{G_bar}
		&G_1(\bm{x})\geq \bar{G}_1(\bm{x}|\bm{x}^p)\\
		&\coloneqq G_1(\bm{x}^p) + 2(\bm{x}^p)^{\mathsf T} \bm{B} (\bm{x}-\bm{x}^p) \notag\\
		&= 2(\bm{x}^p)^{\mathsf T} \bm{B} \bm{x} - G_1(\bm{x}^p). \notag
	\end{align}
	Furthermore, when $\bm{y}$ is fixed, $G_2(\bm{x}, \bm{y})$ is convex w.r.t. $\bm{x}$. Thus, at the $p$th iteration of SCA, a convex surrogate function that provides a global lower-bound for $G(\bm{x}, \bm{y})$ can be constructed as
	\begin{align}
		G(\bm{x},\bm{y}) \geq \bar{G}_1(\bm{x}|\bm{x}^p) - G_2(\bm{x}, \bm{y}).
	\end{align}
	Accordingly, the constraint in \eqref{xBx} can be converted into a convex constraint w.r.t. $\bm{x}$ as
	\begin{align}\label{30}
		\bar{G}_1(\bm{x}|\bm{x}^p) - G_2(\bm{x}, \bm{y}) \geq \delta.
	\end{align}
	
	Next, given $\bm{y}$, the constraint in \eqref{yBy} can be equivalently expressed as
	\begin{align}\label{loglog}
		\frac{\left(\bm{y}^{\mathsf T}\bm{B}\bm{x}\right)^2}{\bm{y}^{\mathsf T}\bm{B}\bm{y}-\delta} \leq \bm{x}^{\mathsf T}\bm{B}\bm{x}.
	\end{align}
	It can be observed that the left-hand side of \eqref{loglog} is jointly convex in $\{\bm{x}, \delta\}$, since $\left(\bm{y}^{\mathsf T} \bm{B} \bm{x}\right)^2$ is a convex quadratic function of $\bm{x}$, while $\bm{y}^{\mathsf T} \bm{B} \bm{y} - \delta$ is linear w.r.t. $\delta$. Moreover, since $G_1(\bm{x}) = \bm{x}^{\mathsf T} \bm{B} \bm{x}$ and is globally lower-bounded by $\bar{G}_1(\bm{x}|\bm{x}^p)$, the constraint in \eqref{loglog} can thus be converted into the following convex constraint:
	\begin{align}\label{36}
		\frac{\left(\bm{y}^{\mathsf T}\bm{B}\bm{x}\right)^2}{\bm{y}^{\mathsf T}\bm{B}\bm{y}-\delta} \leq \bar{G}_1(\bm{x}|\bm{x}^p).
	\end{align}
	
	Thus, at the $p$th iteration, $\bm{x}$ can be optimized by solving the following problem:
	\begin{subequations}
		\begin{align}
			\textrm {(P3)}~~\max_{\bm{x},\bm{V},\delta} \quad & \delta \label{P4a}\\
			\text{s.t.} \quad & \eqref{30}, \eqref{36}, \eqref{P2d}, \eqref{P2e}, \eqref{P2f}. \notag
		\end{align}
	\end{subequations}
	Problem (P3) is a convex optimization problem because all the constraints are convex w.r.t. $\bm{x}$, $\bm{V}$, and $\delta$, allowing it to be efficiently solved using standard convex optimization tools such as CVX \cite{grantcvx}.

	\subsection{Optimization of $\bm{y}$ with Given $\bm{x}$}
	In this subsection, we optimize $\bm{y}$ and $\bm{V}$ in problem (P2) with fixed $\bm{x}$, where the non-convex constraints in \eqref{xBx} and \eqref{yBy} are converted into convex constraints using the SCA approach. Let $\bm{y}^q$ represent the $y$-axis APV at the $q$th SCA iteration. Since the structure of the optimization subproblem of $\bm{y}$ has a similar structure as that of $\bm{x}$, problem (P3) can be modified by substituting $\{\bm{x}, \bm{x}^p, \bm{y}\}$ with $\{\bm{y}, \bm{y}^q, \bm{x}\}$. Accordingly, at the $q$th iteration, the optimization of $\bm{y}$ can be formulated as the following problem:
	\begin{subequations}
		\begin{align}
			\textrm {(P4)}~~\max_{\bm{y},\bm{V},\delta} \quad & \delta \label{P5a}\\
			\text{s.t.} \quad & \bar{G}_1(\bm{y}|\bm{y}^q) - G_2(\bm{y}, \bm{x}) \geq \delta, \label{P5b}\\
			& \frac{\left(\bm{x}^{\mathsf T}\bm{B}\bm{y}\right)^2}{\bm{x}^{\mathsf T}\bm{B}\bm{x}-\delta} \leq \bar{G}_1(\bm{y}|\bm{y}^q), \label{P5c}\\
			& \eqref{P2d}, \eqref{P2e}, \eqref{P2f}, \notag
		\end{align}
	\end{subequations}
	which is a convex optimization problem because all the constraints are convex w.r.t. $\bm{y}$, $\bm{V}$, and $\delta$, allowing it to be efficiently solved using standard convex optimization tools such as CVX \cite{grantcvx}.
	
	\subsection{Overall Algorithm}
	Based on the solutions for problems (P3) and (P4), we present the complete alternating optimization for solving problem (P2) in Algorithm~\ref{alg1}. Specifically, in lines 4–7, $\bm{x}$ is updated iteratively with $\bm{y}$ fixed by solving problem (P3) until the improvement of $\delta$ in \eqref{P4a} falls below the threshold $\epsilon_{\bm{x}}$. Next, in lines 9–12, $\bm{y}$ is updated iteratively with $\bm{x}$ fixed by solving problem (P4) until the improvement of $\delta$ in \eqref{P5a} falls below the threshold $\epsilon_{\bm{y}}$. This alternating optimization continues until the overall increase in $\delta$ falls below a threshold $\epsilon$.
	
	\begin{algorithm}[!t]
		\caption{Alternating Optimization for Problem (P2)}
		\label{alg1}
		\begin{algorithmic}[1]
			\STATE \textbf{Input:} $N$, $\epsilon$, $\epsilon_{\bm{x}}$, $\epsilon_{\bm{y}}$, $\bm{x}^0$, $\bm{y}^0$, $\mathcal{C}$.
			\STATE \textbf{Initialize:} $p \gets 0$, $q \gets 0$, $\bm{x} \gets \bm{x}^0$, $\bm{y} \gets \bm{y}^0$, $\delta \gets 0$.
			
			\WHILE{The increment of $\delta$ is greater than $\epsilon$}
			
			\WHILE{The increment of $\delta$ in \eqref{P4a} is greater than $\epsilon_{\bm{x}}$}
			\STATE Solve problem (P3) to update $\bm{x}^{p+1}$.
			\STATE $p \gets p + 1$.
			\ENDWHILE
			\STATE Update $\bm{x} \gets \bm{x}^p$ and reset $p \gets 0$.
			
			\WHILE{The increment of $\delta$ in \eqref{P5a} is greater than $\epsilon_{\bm{y}}$}
			\STATE Solve problem (P4) to update $\bm{y}^{q+1}$.
			\STATE $q \gets q + 1$.
			\ENDWHILE
			\STATE Update $\bm{y} \gets \bm{y}^q$ and reset $q \gets 0$.
			
			\STATE Update $\delta \gets \min \Big[ {\rm{var}}(\bm{x}) - \frac{{\rm{cov}}(\bm{x},\bm{y})^2}{{\rm{var}}(\bm{y})}, \ {\rm{var}}(\bm{y}) - \frac{{\rm{cov}}(\bm{x},\bm{y})^2}{{\rm{var}}(\bm{x})} \Big]$.
			
			\ENDWHILE
			
			\STATE \textbf{Output:} $\bm{x}$, $\bm{y}$, $\bm{V}$.
		\end{algorithmic}
	\end{algorithm}

	Since the objective function is non-decreasing over iterations and has an upper-bound, the convergence of Algorithm~\ref{alg1} is guaranteed. The computational complexity of solving convex subproblem (P3) or (P4) using the interior-point method is $\mathcal{O}(N^{3.5}\ln(1/\kappa))$, where $\kappa$ denotes the accuracy of interior-point method. Let $I$, $I_{\bm{x}}$, and $I_{\bm{y}}$ represent the maximum number of iterations for the outer alternating loop (lines 4–14), the inner loop for updating $\bm{x}$ (lines 5–6), and the inner loop for updating $\bm{y}$ (lines 10–11), respectively. Consequently, the overall computational complexity of Algorithm~\ref{alg1} can be expressed as $\mathcal{O}\big(N^{3.5}\ln(1/\kappa)(I_{\bm{x}}+I_{\bm{y}})I\big)$.

	\section{Numerical Results}
	This section provides numerical results to assess the performance of the proposed methods for optimizing MA's 1D and 2D trajectories for target AoA estimation. The convergence thresholds in Algorithm~\ref{alg1} are set as $\epsilon=10^{-4}$ and $\epsilon_{\bm{x}}=\epsilon_{\bm{y}}=10^{-2}$. The average received SNR is defined as $P|\beta|^2/\sigma^2$. We set $T_s=10$ µs, $\lambda=0.05$ m, $v^{\rm m} = 10$ m/s, and thus $\Delta=v^{\rm m}T_s=2\times 10^{-3}\lambda$. 
	
	\subsection{1D Antenna Movement}
	In this subsection, we consider the 1D line segment with length $A$ for antenna movement. We set $\theta=45^{\circ}$, and thus $u=\cos \theta=0.71$. We adopt the MLE algorithm to estimate the AoA via \eqref{MLE1D} and compute the corresponding CRB of AoA estimation MSE via \eqref{CRB1D}.
	
	\begin{figure}[!t]
		\centering
		\includegraphics[width=75mm]{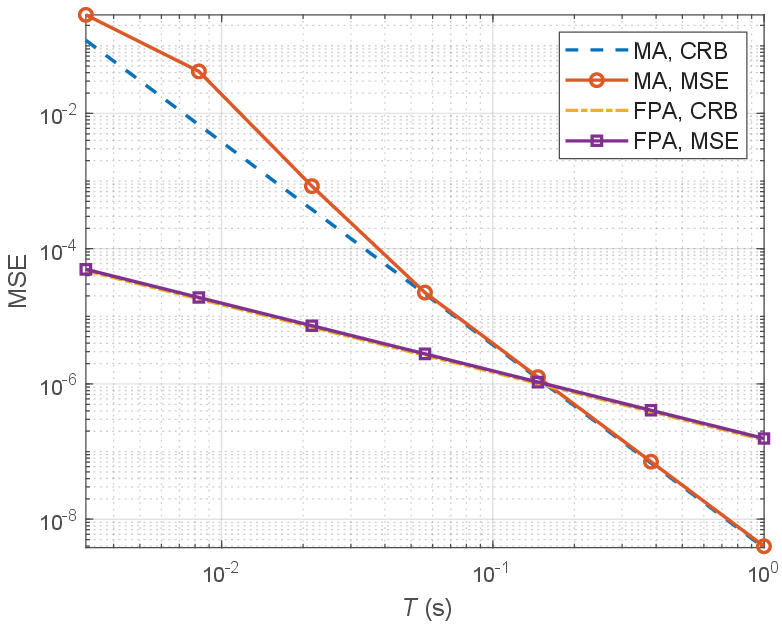}
		\caption{MSE versus sensing time for the single MA-aided and FPA array-aided sensing systems in the case of 1D antenna movement.}
		\label{fig_1d_vs_fpa}
	\end{figure}
	
	First, we consider the case of $A \geq (N-1)\Delta$ (i.e., the TC regime) and compare the proposed single MA-aided sensing system with the conventional FPA array-aided system. The receive SNR is set to $P|\beta|^2/\sigma^2=-20$ dB, and the FPA receiver is equipped with a ULA of $M=16$ antennas. Fig.~\ref{fig_1d_vs_fpa} compares the AoA estimation MSEs (including both the actual MSE and its CRB) versus sensing time for both systems. The results show that the single MA achieves a CRB scaling order of $\mathcal{O}(T^{-3})$, which is substantially faster than the $\mathcal{O}(T^{-1})$ scaling order of the FPA array. Furthermore, the proposed MA scheme outperforms the FPA scheme when $T > \frac{M^{3/2}}{2}\frac{\lambda}{v^{\rm m}} = 0.16$ s, which is easily attainable in practical scenarios.

	\begin{figure}[!t]
		\centering
		\includegraphics[width=75mm]{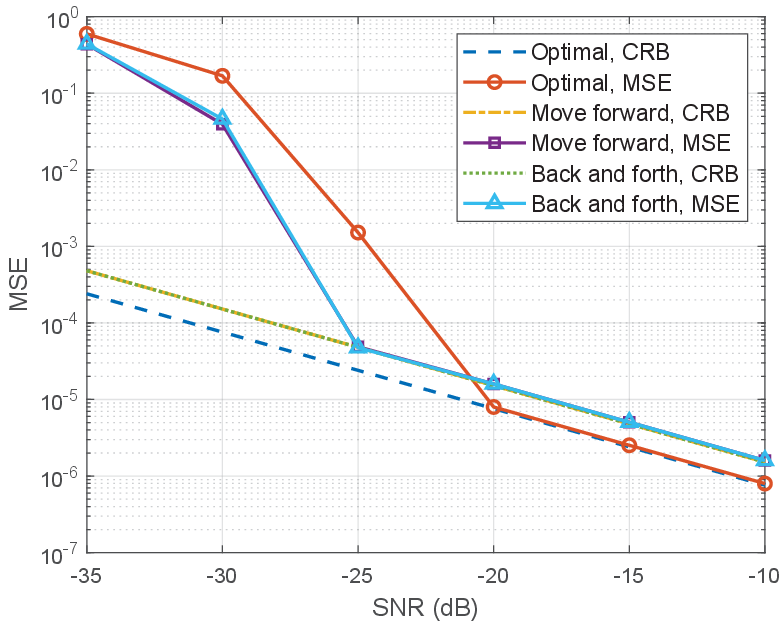}
		\caption{MSE versus SNR for different MA trajectories in the case of 1D antenna movement.}
		\label{1D_SNR}
	\end{figure}
	
	\begin{figure}[!t]
		\centering
		\includegraphics[width=75mm]{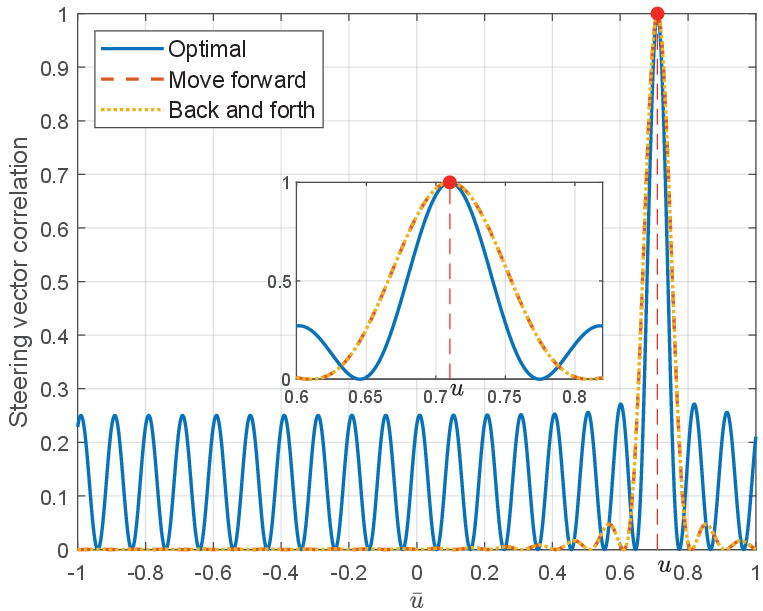}
		\caption{Comparison of steering vector correlation with different MA trajectories in the case of 1D antenna movement.}
		\label{1D_pattern}
	\end{figure}

	Then, we compare the proposed MA trajectory with different benchmark schemes. The considered benchmark schemes for MA trajectory design are listed as follows: 1) \textbf{Move forward}: the MA moves from $x=0$ to $x=A$ at a constant speed of $A/(T_sN)$; and 2) \textbf{Back and forth}: the MA moves back and forth between $x=0$ and $x=A$ at a speed of $v^{\rm m}$. We set $A=10\lambda$, $T=0.1$ s, and thus $N=10^4$. Fig.~\ref{1D_SNR} shows the AoA estimation MSEs versus SNR for different schemes. The results show that the MLE-based AoA estimation curves closely follow the CRB for all schemes. Moreover, the proposed optimal MA trajectory in \eqref{xRL3_main} achieves a substantially lower MSE compared to the benchmark schemes. For example, at $\text{SNR}=-15$ dB, it yields a $50\%$ reduction in MSE compared to both benchmark schemes.
	
	To gain more insights, Fig.~\ref{1D_pattern} illustrates the steering vector correlation for each scheme, defined as $q(\bar{u}|u)\coloneqq\frac{1}{N^2}|\bm{\alpha}(\bm{x},u)^{\mathsf H}\bm{\alpha}(\bm{x},\bar{u})|^2$, with $\bar{u}\in[-1,1]$ and $u=0.71$. Ideally, to maximize spatial resolution and suppress ambiguity in AoA estimation, $q(\bar{u}|u)$ should approach an impulse function, i.e., $q(\bar{u}|u)\to 1$ when $\bar{u}=u$ and $q(\bar{u}|u)\to 0$ otherwise. As shown in Fig.~\ref{1D_pattern}, the proposed MA trajectory generates a much narrower main lobe than the benchmark schemes, thereby providing higher spatial resolution and reduced estimation error. Moreover, the two benchmark schemes exhibit nearly identical correlation patterns, which results in the similar MSE performance as shown in Fig.~\ref{1D_SNR}.

	\subsection{2D Antenna Movement}

		\begin{figure}[!t]
		\centering
		\subfigure[${\rm{MSE}}(u)$]{
			\begin{minipage}{.47\textwidth}
				\centering
				\includegraphics[scale=.54]{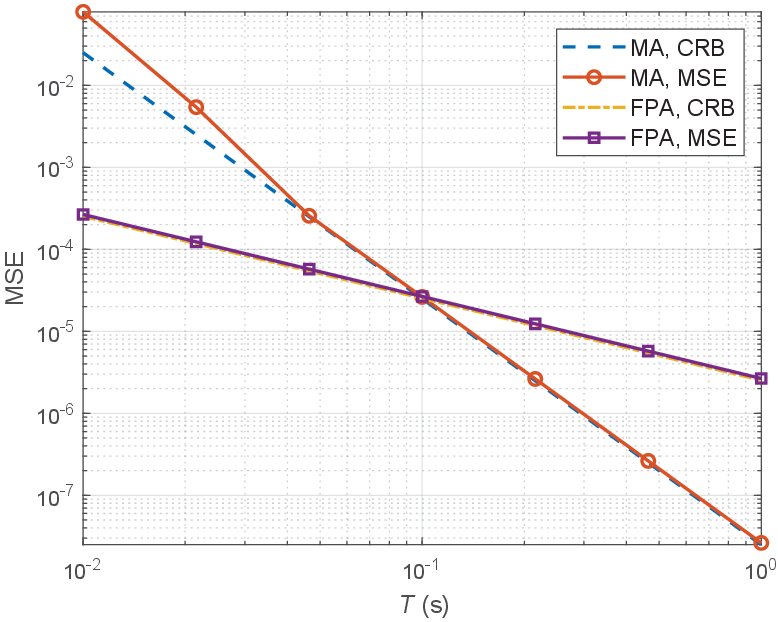}
			\end{minipage}
			\label{2D_versus_FPA_u}
		}
		\subfigure[${\rm{MSE}}(v)$]{
			\begin{minipage}{.47\textwidth}
				\centering
				\includegraphics[scale=.54]{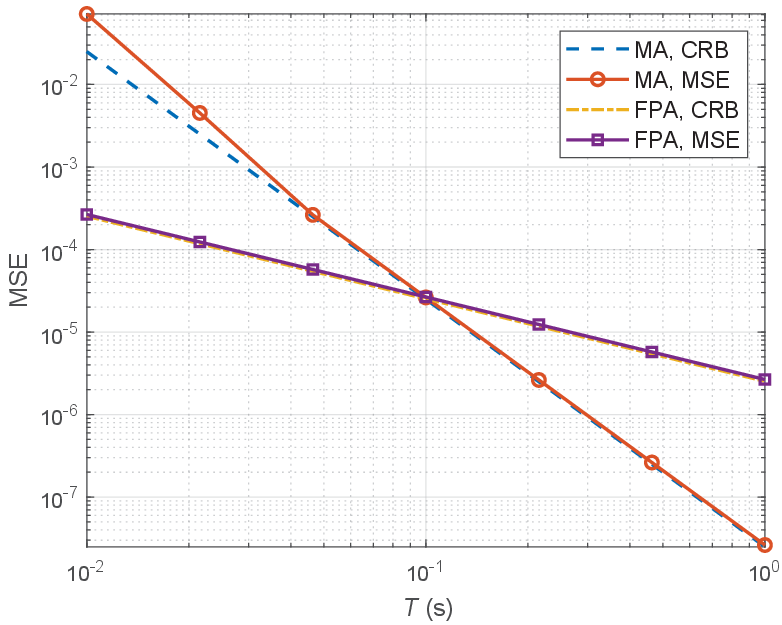}
			\end{minipage}
			\label{2D_versus_FPA_v}
		}
		\caption{MSE versus sensing time for the single MA-aided and FPA array-aided sensing systems in the case of 2D antenna movement.}
		\label{FIG_2D_versus_FPA}
	\end{figure}
	
	\begin{figure}[!t]
		\centering
		\subfigure[Without the region size constraint]{
			\begin{minipage}{.47\textwidth}
				\centering
				\includegraphics[scale=.54]{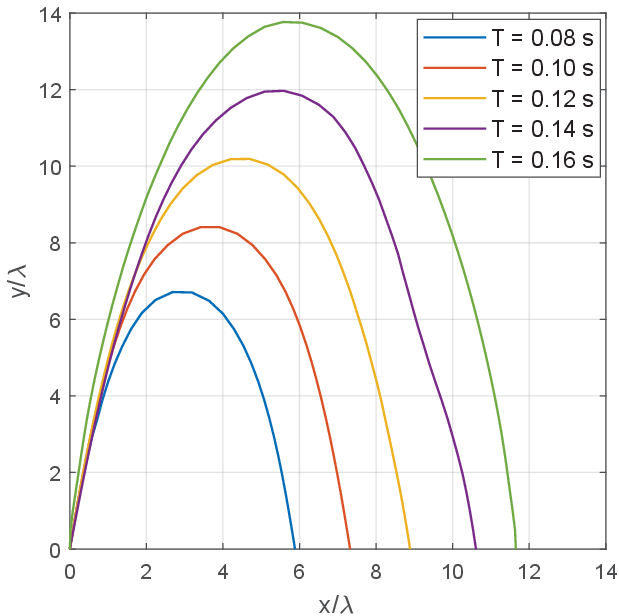}
			\end{minipage}
			\label{2D_position_36}
		}
		\subfigure[With the region size constraint]{
			\begin{minipage}{.47\textwidth}
				\centering
				\includegraphics[scale=.54]{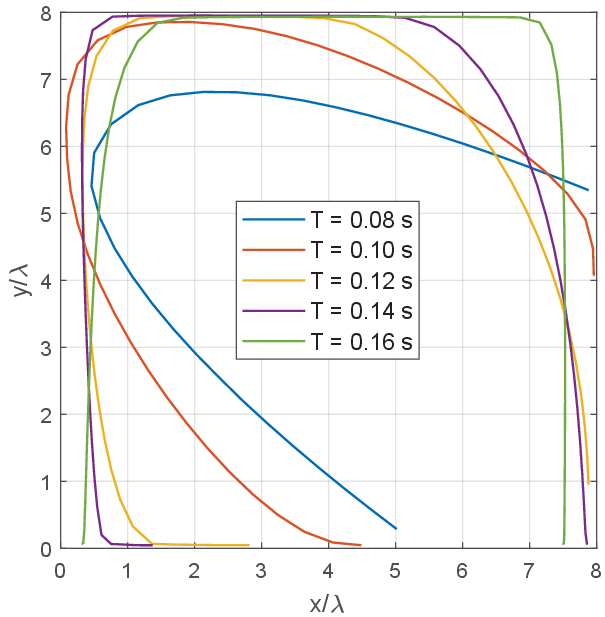}
			\end{minipage}
			\label{2D_position_100}
		}
		\caption{Illustration of the MA trajectories for 2D antenna movement.}
		\label{FIG4}
	\end{figure}
	
	\begin{figure}[!t]
		\centering
		\subfigure[${\rm{MSE}}(u)$]{
			\begin{minipage}{.47\textwidth}
				\centering
				\includegraphics[scale=.54]{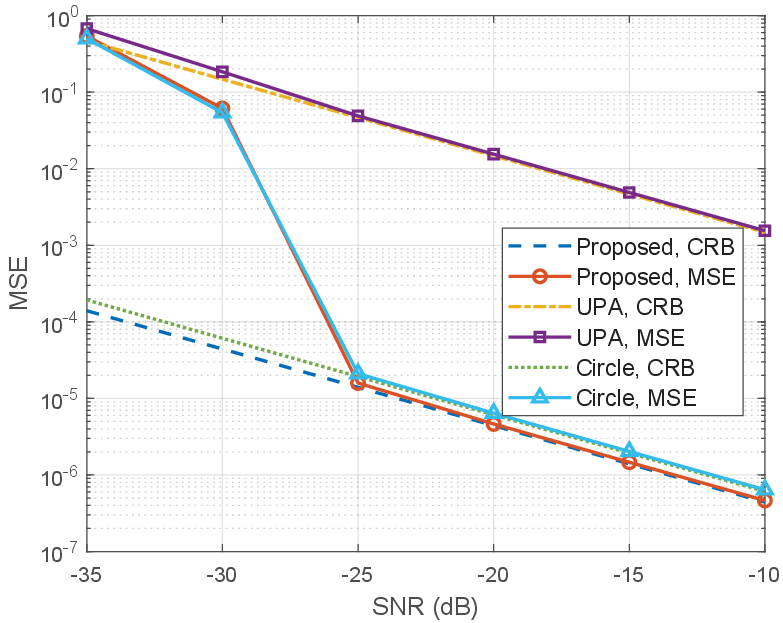}
			\end{minipage}
			\label{2D_SNR_u}
		}
		\subfigure[${\rm{MSE}}(v)$]{
			\begin{minipage}{.47\textwidth}
				\centering
				\includegraphics[scale=.54]{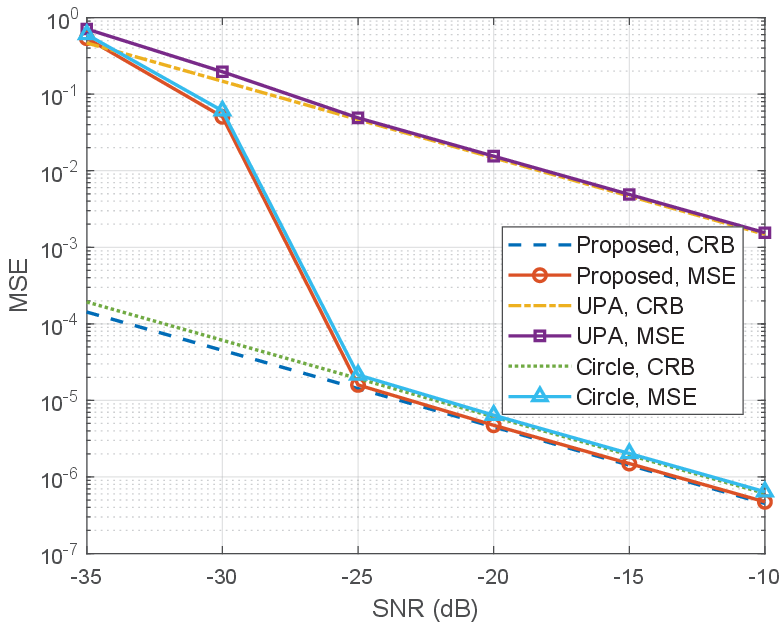}
			\end{minipage}
			\label{2D_SNR_v}
		}
		\caption{MSE versus SNR for different MA trajectories in the case of 2D antenna movement.}
		\label{FIG5}
	\end{figure}
	
	\begin{figure}[!t]
		\centering
		\subfigure[Proposed]{
			\begin{minipage}{.47\textwidth}
				\centering
				\includegraphics[scale=.54]{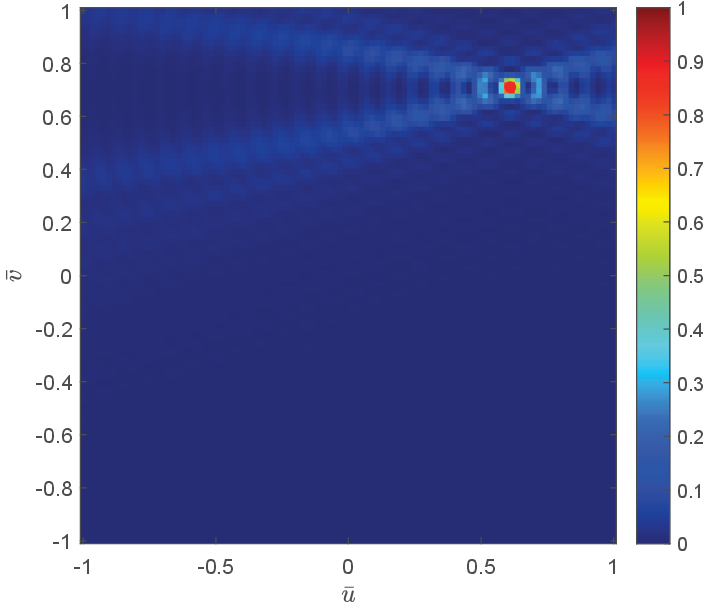}
			\end{minipage}
			\label{2D_correlation_proposed}
		}
		\subfigure[UPA]{
			\begin{minipage}{.47\textwidth}
				\centering
				\includegraphics[scale=.54]{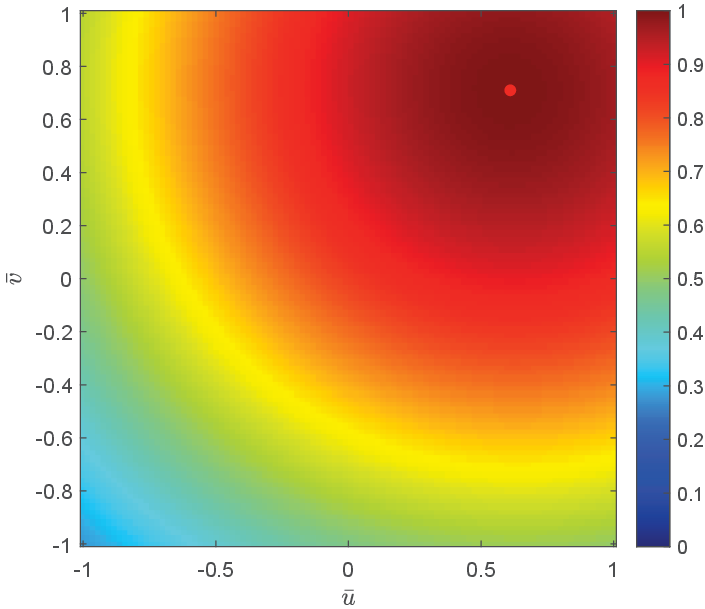}
			\end{minipage}
			\label{2D_correlation_UPA}
		}
		\subfigure[Circle]{
			\begin{minipage}{.47\textwidth}
				\centering
				\includegraphics[scale=.54]{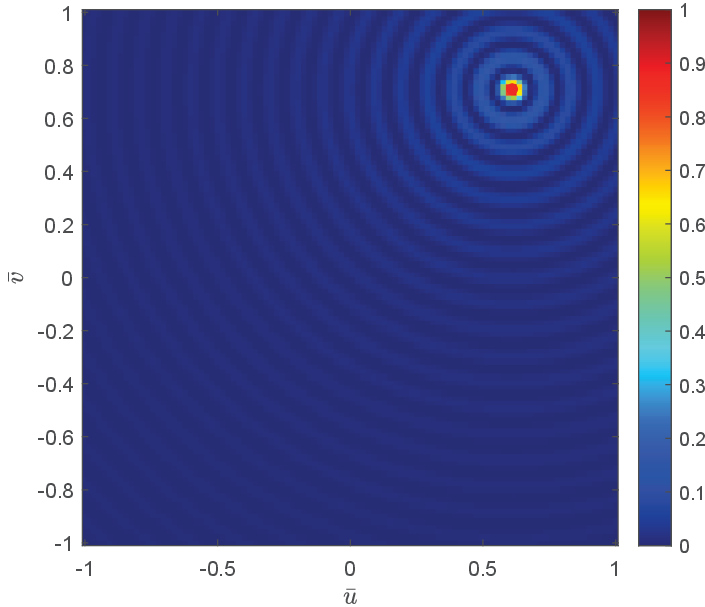}
			\end{minipage}
			\label{2D_correlation_circle}
		}
		\caption{Comparison of steering vector correlation with different MA trajectories in the case of 2D antenna movement.}
		\label{FIG6}
	\end{figure}

	In this subsection, numerical results are provided to verify the performance of Algorithm~\ref{alg1} for AoA estimation in the case of 2D antenna movement. We consider the 2D square region $\mathcal{C}^\textrm{squ}$ with size $A\times A$. We set $\theta=45^{\circ}$ and $\phi=30^{\circ}$, thus $u = \sin\theta \cos\phi=0.61$ and $v = \cos\theta=0.71$. We adopt the MLE algorithm to estimate the spatial AoAs via \eqref{MLE2D} and compute the corresponding CRB of AoA estimation MSE via \eqref{CRBr}.
	
	First, we consider the MA moves at the maximum speed $v^{\rm m}$ along a circular trajectory with radius $R^{\rm cir} = \frac{\Delta}{2\sin(\tfrac{\pi}{N})}$, and compare the proposed single MA-aided sensing system with the conventional 2D FPA array-aided system. The receive SNR is set to $P|\beta|^2/\sigma^2=-20$ dB, and the FPA receiver is equipped with a UPA of $M=4\times 4$ antennas. Fig.~\ref{FIG_2D_versus_FPA} compares the AoA estimation MSEs versus sensing time for both systems. The results show that the single MA achieves a CRB scaling order of $\mathcal{O}(T^{-3})$, which is substantially faster than the $\mathcal{O}(T^{-1})$ scaling order of the FPA array. Furthermore, the proposed MA scheme outperforms the FPA scheme when $T > \frac{\pi M}{\sqrt{6}} \frac{\lambda}{v^{\rm m}} = 0.1$ s, which is even smaller than its previous counterpart in the case of 1D antenna movement.

	Next, we compare the proposed MA trajectory with different benchmark schemes. The considered benchmark schemes for MA trajectory design are listed as follows: 1) \textbf{UPA}: the MA moves at speed $v^{\rm m}$ to form a $\sqrt{N}\times \sqrt{N}$ UPA, with inter-position spacing $\Delta$ in both the horizontal and vertical directions; and 2) \textbf{Circle}: the MA moves at speed $v^{\rm m}$ to form a circular trajectory with radius $R^{\rm cir} = \tfrac{\Delta}{2\sin(\tfrac{\pi}{N})}$.	In Fig.~\ref{FIG4}, we illustrate the MA trajectories for the proposed Algorithm~\ref{alg1}. To facilitate the optimization, every $250$ adjacent elements in the AVM $\bm{V}$ share the same value. In Fig.~\ref{2D_position_36}, we set $\bm{r}_1=[0,0]^{\mathsf T}$ and assume that the antenna movement region is sufficiently large, such that the constraint \eqref{P2e} is inactive. By contrast, in Fig.~\ref{2D_position_100}, we consider a finite region size with $A=8\lambda$. It can be observed from Fig.~\ref{FIG4} that the MA’s trajectories are significantly affected by the region size constraint. Specifically, Fig.~\ref{2D_position_36} demonstrates that without the constraint \eqref{P2e}, the trajectories expand widely in the $x$-axis and $y$-axis directions as $T$ increases, thereby forming a large virtual aperture via antenna movement and enhancing angular resolution for improved AoA estimation. While Fig.~\ref{2D_position_100} shows that under the region size constraint, the trajectories are confined within the predefined region, compelling the MAs to move along the boundaries when $T$ is sufficiently large.
	
	Fig.~\ref{FIG5} shows the AoA estimation MSEs versus SNR for different schemes. We set $A=15\lambda$ and $T=0.16$ s. The results show that the MLE-based AoA estimation curves closely follow the CRB for all schemes, and the proposed MA trajectory achieves a substantially lower MSE compared to the benchmark schemes. Moreover, the proposed scheme yields similar performance of ${\rm{MSE}}(u)$ and ${\rm{MSE}}(v)$, demonstrating its effectiveness in jointly estimating $u$ and $v$.
	
	To gain more insights, Fig.~\ref{FIG6} illustrates the steering vector correlation for each scheme, defined as $\frac{1}{N^2}|\bm{\alpha}(\bm{R},\bm{\eta})^{\mathsf H} \bm{\alpha}(\bm{R},\bar{\bm{\eta}})|^2$, with $\bar{\bm{\eta}}\in[-1,1]\times[-1,1]$ and $\bm{\eta}=[0.61,0.71]^{\mathsf T}$. We set $A=15\lambda$ and $T=0.16$ s. As shown in Fig.~\ref{2D_position_36}, the proposed MA trajectory forms a virtual aperture of $12\lambda \times 14\lambda$. In contrast, the circular trajectory with radius $R^{\rm cir} = \tfrac{\Delta}{2\sin(\tfrac{\pi}{N})}=5\lambda$ forms a virtual aperture of $2R^{\rm cir} \times 2R^{\rm cir} = 10\lambda \times 10\lambda$, whereas the UPA trajectory only achieves a virtual aperture of $\sqrt{N}\Delta \times \sqrt{N}\Delta = 0.25\lambda \times 0.25\lambda$, which is much smaller than both the proposed and circular trajectories. Consequently, as shown in Fig.~\ref{FIG6}, the proposed MA trajectory generates a much narrower main lobe than the benchmark schemes, thereby providing higher spatial resolution and reduced AoA estimation error.

	\section{Conclusion}
	In this paper, we proposed a new wireless sensing system employing a single MA to improve the AoA estimation accuracy by optimizing the antenna's trajectory over time. We showed that the angle estimation performance is fundamentally determined by the MA trajectory, and derived the CRB of the MSE for AoA estimation as a function of the trajectory for both 1D and 2D antenna movement. For the 1D case, a globally optimal trajectory for minimizing the CRB was obtained in closed form. Notably, the resulting CRB decreases cubically with sensing time in the TC regime, whereas it decreases linearly with sensing time and quadratically with the movement line segment's length in the SC regime. For the 2D case, we aimed to achieve the min-max CRBs of estimation MSE for the two AoAs w.r.t. the horizontal and vertical axes. To this end, we designed an efficient alternating optimization algorithm that iteratively updates the MA’s horizontal or vertical coordinates with the other being fixed, yielding a locally optimal trajectory. Numerical results showed that the proposed 1D/2D MA-based sensing schemes significantly reduce both the CRB and actual AoA estimation MSE compared to conventional FPA-based sensing with ULAs/UPAs as well as various benchmark MA trajectories. Moreover, it was revealed that the steering vectors of our designed 1D/2D MA trajectories have low correlation in the angular domain, thereby effectively increasing the angular resolution for achieving higher AoA estimation accuracy.

	\begin{figure}[!t]
		\centering
		\includegraphics[width=85mm]{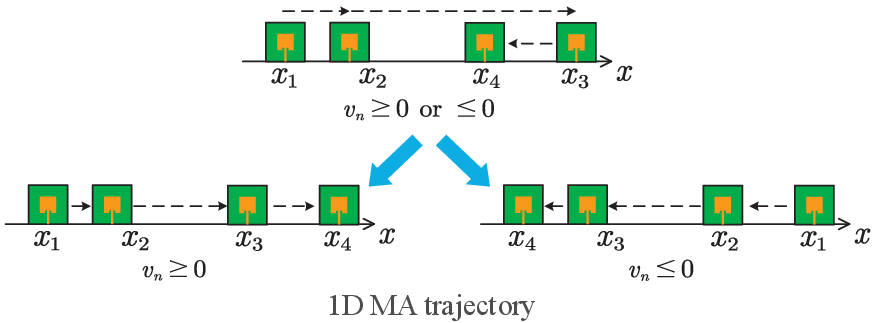}
		\caption{Illustration of the equivalent trajectory with $v_n \geq0$ or  $v_n \leq0$.}
		\label{trajectory_equivalence}
	\end{figure}
	
	\appendix
	\subsection{Proof of Theorem 1}
	As illustrated in Fig.~\ref{trajectory_equivalence}, for any 1D MA trajectory with velocity satisfying $-v^{\rm m} \le v_n \le v^{\rm m}$, $n=1, \dots, N-1$, an equivalent trajectory can always be constructed that achieves the same APV while restricting the velocity to $0 \le v_n \le v^{\rm m}$ or $-v^{\rm m} \le v_n \le 0$. Therefore, without loss of generality, we assume $v_n \geq 0$ and $x_1 \leq x_2 \leq \cdots \leq x_N$. The counterpart for $v_n \leq 0$ can be obtained by mirroring the APV and reversing the velocity directions. Specifically, for $(\bm{x},\bm{v})$ with $v_n \geq 0$, the mirrored solution $(\tilde{\bm{x}},\tilde{\bm{v}})$ is given by
	\begin{align}
		\tilde{x}_n &= A - x_{N+1-n}, \quad n=1,\ldots,N, \notag\\
		\tilde{v}_n &= -v_{N-n}, \quad n=1,\ldots,N-1.
	\end{align}
	Then, from \eqref{P1d}, it follows that $x_n - x_{n-1} \leq \Delta$, $n=2, \dots, N$. The objective function of problem (P1) can then be expressed as
	\begin{align}
		\text{var}(\bm{x}) &= \frac{1}{N^2} \sum_{n=2}^{N} \sum_{m=1}^{n-1} (x_n - x_m)^2 \notag\\
		&\leq \frac{1}{N^2} \sum_{n=2}^{N} \sum_{m=1}^{n-1} (n-m)^2 \Delta^2 = \frac{(N^2-1)\Delta^2}{12},
	\end{align}
	where the upper-bound is achieved when $x_n - x_{n-1} = \Delta$, $n=2, \dots, N$.
	
	Therefore, if $A \geq (N-1)\Delta$, such that $x_n - x_{n-1} = \Delta$ for all $n=2, \dots, N$, the optimal AVV $\bm{v}^{\rm TC}$ and APV $\bm{x}^{\rm TC}$ for problem (P1) are given by \eqref{vTC} and \eqref{xTC}, respectively,
	where $x_1^{\rm TC}$ satisfies $0 \leq x_1^{\rm TC} \leq A - (N-1)\Delta$
	to ensure $0 \leq x_1^{\rm TC} \leq x_2^{\rm TC} \leq \cdots \leq x_N^{\rm TC} \leq A$. The corresponding minimum CRB is then given by
	\begin{align}\label{crbinf}
		{\rm CRB}_u\left(\bm{x}^{\rm TC}\right) = \frac{3\sigma^2\lambda^2}{2\pi^2P|\beta|^2T_s^2}  \frac{1}{N(N^2-1)(v^{\rm m})^2}.
	\end{align}
	It shows that ${\rm CRB}_u\left(\bm{x}^{\rm TC}\right)$ decreases with $N$ in the order of $\mathcal{O}(N^{-3})$, while it decreases with $v^{\rm m}$ in the order of $\mathcal{O}((v^{\rm m})^{-2})$. This completes the proof of Theorem~1.
	
	\subsection{Proof of Theorem 2}
	Similar to the proof of Theorem~1, we assume $v_n \geq 0$ and $x_1 \leq x_2 \leq \cdots \leq x_N$ without loss of generality. The counterpart for $v_n \leq 0$ can be obtained by mirroring the APV and reversing the velocity directions. Define $\bm{\bar{x}}_n \coloneqq [x_1,\ldots,x_{n-1},x_{n+1},\ldots,x_N]^{\mathsf T} \in \mathbb{C}^{(N-1)\times1}$, $n=1,\dots, N$. The objective function of problem (P1) can then be expressed as a function of $x_n$:
	\begin{align}\label{var_xn}
		\text{var}(x_n|\bm{\bar{x}}_n) = \frac{N-1}{N^2}\left( \left( x_n - \mu(\bm{\bar{x}}_n) \right)^2 + N\text{var}(\bm{\bar{x}}_n) \right).
	\end{align}
	From \eqref{var_xn}, it is shown that $\text{var}(x_n|\bm{\bar{x}}_n)$ is a convex quadratic function of $x_n$, with the feasible region given by $x_n\in [x_{n,{\rm L}}, x_{n,{\rm R}}]	\coloneqq\left[\max\left( x_{n+1}-\Delta,x_{n-1} \right), \min\left( x_{n-1}+\Delta,x_{n+1} \right) \right]$.	This implies that, once $\bm{\bar{x}}_n$ is fixed, the maximum of $\text{var}(x_n|\bm{\bar{x}}_n)$ must be attained at the boundary of the feasible region of $x_n$, i.e., $x_n = x_{n,{\rm L}}$ or $x_n = x_{n,{\rm R}}$.
	
	Next, we show that if there exists $\tilde{n} \in \{1, \ldots, N\}$ such that $\arg \max_{x_{\tilde{n}}} \text{var}(x_{\tilde{n}}|\bm{\bar{x}}_{\tilde{n}}) = x_{\tilde{n},{\rm R}}$, then $\arg \max_{x_n} \text{var}(x_n|\bm{\bar{x}}_n) \equiv x_{n,{\rm R}}$, $n=\tilde{n}+1,\ldots,N$. 
	Specifically, from \eqref{var_xn}, $\text{var}(x_{\tilde{n}}|\bm{\bar{x}}_{\tilde{n}})$ is a convex quadratic function of $x_{\tilde{n}} \in [x_{\tilde{n},{\rm L}}, x_{\tilde{n},{\rm R}}]$ with axis of symmetry at $x = \mu(\bm{\bar{x}}_{\tilde{n}})$. Since $\text{var}(x_{\tilde{n},{\rm R}}|\bm{\bar{x}}_{\tilde{n}}) \geq \text{var}(x_{\tilde{n},{\rm L}}|\bm{\bar{x}}_{\tilde{n}})$, it follows that
	\begin{align}\label{xxmu}
		x_{\tilde{n},{\rm R}} + x_{\tilde{n},{\rm L}} \geq 2\mu(\bm{\bar{x}}_{\tilde{n}}).
	\end{align}
	We first prove that $\mu(\bm{\bar{x}}_n)$ is a non-increasing function of $x_n$ for $n = 1, \ldots, N$. Specifically, the difference between $\mu(\bm{\bar{x}}_{n+1})$ and $\mu(\bm{\bar{x}}_n)$ ($n=1, \ldots, N-1$) is $\mu(\bm{\bar{x}}_{n+1}) - \mu(\bm{\bar{x}}_n) = \frac{1}{N-1} (x_n - x_{n+1})$.
	Since $x_n \le x_{n+1}$, it follows that $\mu(\bm{\bar{x}}_{n+1}) - \mu(\bm{\bar{x}}_n) \le 0$, showing that $\mu(\bm{\bar{x}}_n)$ is a non-increasing function of $x_n$ for $n = 1, \ldots, N$. On the other hand, since $x_n \ge x_{n-1}$ and $x_{n+2} \ge x_{n+1}$, we have $x_n + \Delta \ge x_{n-1} + \Delta$ and $x_{n+2} \ge x_{n+1}$. Since the $\min(\cdot)$ function is non-decreasing w.r.t. its arguments, it follows that $\min(x_n + \Delta, x_{n+2}) \ge \min(x_{n-1} + \Delta, x_{n+1})$,
	which implies $x_{n+1,{\rm R}} \ge x_{n,{\rm R}}$. Similarly, we have $x_{n+1,{\rm L}} \ge x_{n,{\rm L}}$. Since $x_{n,{\rm R}}$ and $x_{n,{\rm L}}$ are non-decreasing while $\mu(\bm{\bar{x}}_n)$ is non-increasing of $x_n$ for $n = 1, \ldots, N$, inequality \eqref{xxmu} always holds. Hence, we have $\arg \max_{x_n} \text{var}(x_n|\bm{\bar{x}}_n) \equiv x_{n,{\rm R}}$, $n=\tilde{n},\ldots,N$.
	Consequently, we have $	x_n = x_{n-1} + \Delta$ or $x_n = x_{n+1}$, $n = \tilde{n}, \ldots, N$.
	If there exists $N-N_{\rm R}+1 \in \{\tilde{n}, \ldots, N\}$ such that $x_n = x_{n+1}$, then, since $x_{n+1} \neq x_n + \Delta$, it follows that $x_{n+1} = x_{n+2}$. By induction, we have $x_{N-N_{\rm R}+1} = x_{N-N_{\rm R}+2} = \cdots = x_N$.
	Consequently, $x_n$ ($n = \tilde{n}, \ldots, N$) is given by
	\begin{align}\label{xR}
		x_n= \left\{\begin{array}{l}
			x_{\tilde{n}} + (n-{\tilde{n}})\Delta,\ n=\tilde{n},\ldots,N-N_{\rm R}, \\
			x_{\rm max},\ n=N-N_{\rm R}+1,\ldots,N, \end{array}\right.
	\end{align}
	where $x_{\rm max} = \max\limits_{n=\tilde{n},\ldots,N}  x_n$.
	
	Since $\text{var}(x_n|\bm{\bar{x}}_n)$ is symmetric w.r.t. $x = \mu(\bm{\bar{x}}_n)$, it follows analogously that if there exists $\hat{n} \in \{1, \ldots, N\}$ such that $\arg \max_{x_{\hat{n}}} \text{var}(x_{\hat{n}}|\bm{\bar{x}}_{\hat{n}}) = x_{\hat{n},{\rm L}}$,
	then, we have $\arg \max_{x_n} \text{var}(x_n|\bm{\bar{x}}_n) \equiv x_{n,{\rm L}}$, $n = 1, \ldots, \hat{n}$.
	Hence, we have $x_n = x_{n+1} - \Delta$ or $x_n = x_{n-1}$, $n = 1, \ldots, \hat{n}$.
	If there exists $N_{\rm L} \in \{1, \ldots, \hat{n}\}$ such that $x_n = x_{n-1}$, then, since $x_{n-1} \neq x_n - \Delta$, it follows that $x_{n-1} = x_{n-2}$. By induction, we obtain $x_1 = x_2 = \cdots = x_{N_{\rm L}}$.
	Consequently, $x_n$ ($n = 1, \ldots, \hat{n}$) is given by
	\begin{align}\label{xL}
		x_n= \left\{\begin{array}{l}
			x_{\rm min},\ n=1,\ldots,N_{\rm L}, \\
			x_{\hat{n}} + (n-{\hat{n}})\Delta,\ n=N_{\rm L}+1,\ldots,\hat{n}, \end{array}\right.
	\end{align}
	where $x_{\rm min} = \min\limits_{n=1,\ldots,\hat{n}}  x_n$.
	
	Therefore, by combining the results in \eqref{xL} and \eqref{xR}, we have $\tilde{n} = \hat{n}+1$, and the optimal 1D MA trajectory for problem (P1) can be expressed as
	\begin{align}
		x_n= \left\{\begin{array}{l}
			x_{\rm min},\ n=1,\ldots,N_{\rm L}, \\
			x_{\hat{n}} + (n-{\hat{n}})\Delta,\ n=N_{\rm L}+1,\ldots,\hat{n}, \\
			x_{\hat{n}+1} + (n-\hat{n}-1)\Delta,\ n=\hat{n}+1,\ldots,N-N_{\rm R}, \\
			x_{\rm max},\ n=N-N_{\rm R}+1,\ldots,N. \end{array}\right. \notag
	\end{align}
	
	Next, we prove by contradiction that $x_{\hat{n}+1} - x_{\hat{n}} = \Delta$. Specifically, suppose $\bm{x}$ is the optimal solution for maximizing $\text{var}(\bm{x})$ with $x_{\hat{n}+1} - x_{\hat{n}} < \Delta$. If $x_{\hat{n}+1} > \mu(\bm{\bar{x}}_{\hat{n}+1})$, then, according to \eqref{var_xn}, setting $x_{\hat{n}+1} \leftarrow x_{\hat{n}} + \Delta$ increases $\left( x_{\hat{n}+1} - \mu(\bm{\bar{x}}_{\hat{n}+1}) \right)^2$, thereby increasing $\text{var}(x_{\hat{n}+1}|\bm{\bar{x}}_{\hat{n}+1})$. Otherwise, if $x_{\hat{n}+1} \le \mu(\bm{\bar{x}}_{\hat{n}+1})$, setting $x_{\hat{n}} \leftarrow x_{\hat{n}+1} - \Delta$ increases $\left( x_{\hat{n}} - \mu(\bm{\bar{x}}_{\hat{n}}) \right)^2$, thereby increasing $\text{var}(x_{\hat{n}}|\bm{\bar{x}}_{\hat{n}})$. In either case, $\text{var}(\bm{x})$ is increased, contradicting the assumption that $\bm{x}$ is optimal. Therefore, we must have $x_{\hat{n}+1} - x_{\hat{n}} = \Delta$.
	Consequently, the optimal 1D MA trajectory for problem (P1) takes the form of
	\begin{align}
		x_n= \left\{\begin{array}{l}
			x_{\rm min},\ 1\leq n\leq N_{\rm L}, \\
			x_{N_{\rm L}+1} + (n-N_{\rm L}-1)\Delta,\ N_{\rm L}+1\leq n\leq N-N_{\rm R}, \\
			x_{\rm max},\ N-N_{\rm R}+1\leq n\leq N. \end{array}\right. \notag
	\end{align}
	
	Subsequently, we prove by contradiction that $x_{\rm max} = A$ and $x_{\rm min} = 0$. Suppose $\bm{x}$ is the optimal solution for maximizing $\text{var}(\bm{x})$ with $x_{\rm max} < A$. Since $\mu(\bm{\bar{x}}_N) < x_{\rm max}$, then according to \eqref{var_xn}, setting $x_N \leftarrow \min(x_{N-1} + \Delta, A) = \min(x_{\rm max} + \Delta, A)$
	increases $\left(x_N - \mu(\bm{\bar{x}}_N)\right)^2$ and hence increases $\text{var}(x_N|\bm{\bar{x}}_N)$. This contradicts the assumption that $\bm{x}$ is optimal. Therefore, we must have $x_{\rm max} = A$. Similarly, by a symmetric argument, we can show that $x_{\rm min} = 0$ in a similar manner.
	
	We now prove that the optimal solution to problem (P1) satisfies $x_{N_{\rm L}+1} = \Delta$ or $x_{N-N_{\rm R}} = A - \Delta$. Let $N_{\rm M} \coloneqq N - N_{\rm L} - N_{\rm R}$, and define $\bm{x}\coloneqq[\bm{x}_{\rm L}^{\mathsf T}, \bm{x}_{\rm M}^{\mathsf T}, \bm{x}_{\rm R}^{\mathsf T}]^{\mathsf T}$,
	where $\bm{x}_{\rm L} \coloneqq [x_1,\ldots,x_{N_{\rm L}}]^{\mathsf T} \in \mathbb{R}^{N_{\rm L}\times1}$,
	$\bm{x}_{\rm M} \coloneqq [x_{N_{\rm L}+1},\ldots,x_{N-N_{\rm R}}]^{\mathsf T} \in \mathbb{R}^{N_{\rm M}\times1}$, and
	$\bm{x}_{\rm R} \coloneqq [x_{N-N_{\rm R}+1},\ldots,x_N]^{\mathsf T} \in \mathbb{R}^{N_{\rm R}\times1}$. For ease of notation, we further define $\bm{x}_{\rm LR}\coloneqq[\bm{x}_{\rm L}^{\mathsf T},\bm{x}_{\rm R}^{\mathsf T}]^{\mathsf T}\in\mathbb{R}^{(N_{\rm L}+N_{\rm R})\times1}$.
	Given $\bm{x}_{\rm LR}$ and $N_{\rm M}$, $\text{var}(\bm{x})$ can then be expressed as a function of $x_{N_{\rm L}+1}$ as $\text{var}(x_{N_{\rm L}+1}|\bm{x}_{\rm LR},N_{\rm M}) = \frac{(N_{\rm L}+N_{\rm R})N_{\rm M}}{N^2} \left( \mu(\bm{x}_{\rm M})-\mu(\bm{x}_{\rm LR}) \right)^2 + \frac{N_{\rm L}+N_{\rm R}}{N}\text{var}(\bm{x}_{\rm LR}) + \frac{N_{\rm M}}{N}\text{var}(\bm{x}_{\rm M})$.
	Since $\bm{x}_{\rm M}$ forms a ULA with antenna spacing $\Delta$, its variance $\text{var}(\bm{x}_{\rm M})$ is independent of its first element $x_{N_{\rm L}+1}$. Moreover, the mean of $\bm{x}_{\rm M}$ can be expressed as $\mu(\bm{x}_{\rm M}) = \frac{1}{2}\left(x_{N_{\rm L}+1} + x_{N-N_{\rm R}}\right) = x_{N_{\rm L}+1} + \frac{N_{\rm M}-1}{2}\Delta$,
	which is linear with $x_{N_{\rm L}+1}$. Furthermore, we have $x_{N_{\rm L}+1}\in(0,\Delta]$ and $A - x_{N-N_{\rm R}} = A-(x_{N_{\rm L}+1}+(N_{\rm M}-1)\Delta) \in(0,\Delta]$,
	which implies that $x_{N_{\rm L}+1}$ belongs to the feasible domain $\mathcal{D}(x_{N_{\rm L}+1})\coloneqq (0,\Delta] \cup [A-N_{\rm M}\Delta,A-(N_{\rm M}-1)\Delta) = \left\{\begin{array}{l}
		(0,A-(N_{\rm M}-1)\Delta),\ A-N_{\rm M}\Delta\leq0, \\
		{[A-N_{\rm M}\Delta,\Delta]},\ A-N_{\rm M}\Delta>0. \end{array}\right.$
	Since $\text{var}(\bm{x}_{\rm M}|\bm{x}_{\rm LR})$ is a convex quadratic function of $x_{N_{\rm L}+1}\in \mathcal{D}(x_{N_{\rm L}+1})$, its maximum must be attained at the boundary of the feasible region $\mathcal{D}(x_{N_{\rm L}+1})$. If $A-N_{\rm M}\Delta \leq 0$ such that $\mathcal{D}(x_{N_{\rm L}+1})=(0, A-(N_{\rm M}-1)\Delta)$, there always exists $\bar{\tau}>0$ with $x_{N_{\rm L}+1}\pm\bar{\tau} \in \mathcal{D}(x_{N_{\rm L}+1})$. In this case, we have $\text{var}(x_{N_{\rm L}+1}|\bm{x}_{\rm LR},N_{\rm M}) \leq  \max(\text{var}(x_{N_{\rm L}+1}+\bar{\tau}|\bm{x}_{\rm LR},N_{\rm M}), \text{var}(x_{N_{\rm L}+1}-\bar{\tau}|\bm{x}_{\rm LR},N_{\rm M}))$,
	which contradicts the assumption that $\bm{x}$ is the optimal solution for maximizing $\text{var}(\bm{x})$. Therefore, we must have $A-N_{\rm M}\Delta>0$,
	and the optimal solution satisfies $x_{N_{\rm L}+1} = A-N_{\rm M}\Delta$ or $x_{N_{\rm L}+1} = \Delta$,
	i.e., $x_{N-N_{\rm R}} = x_{N_{\rm L}+1}+(N_{\rm M}-1)\Delta = A-\Delta$ or $x_{N_{\rm L}+1} = \Delta$.
	Since the cases $x_{N-N_{\rm R}}=A-\Delta$ and $x_{N_{\rm L}+1}=\Delta$ are symmetric w.r.t. $x=\frac{A}{2}$, we only consider $x_{N_{\rm L}+1}=\Delta$ in the following for simplicity. Then, we can determine $N_{\rm M}$ as
	\begin{align}\label{NM}
		N_{\rm M} &= \frac{1}{\Delta}\left(x_{N-N_{\rm R}} - x_{N_{\rm L}+1}\right) + 1 = \left\lceil\frac{A}{\Delta}\right\rceil-1.
	\end{align}
	
	Finally, we determine the values of $N_{\rm L}$ and $N_{\rm R}$. Given $\bm{x}_{\rm M}$, the variance $\text{var}(\bm{x})$ can be expressed as a function of $N_{\rm L}$ as $\text{var}(N_{\rm L}|\bm{x}_{\rm M}) = \frac{N_{\rm L}}{N}\text{var}(\bm{x}_{\rm L}) + \frac{N_{\rm M}}{N}\text{var}(\bm{x}_{\rm M}) + \frac{N_{\rm R}}{N}\text{var}(\bm{x}_{\rm R})  + \frac{N_{\rm L}}{N}(\mu(\bm{x}_{\rm L})-\mu(\bm{x}))^2 + \frac{N_{\rm M}}{N}(\mu(\bm{x}_{\rm M})-\mu(\bm{x}))^2  + \frac{N_{\rm R}}{N}(\mu(\bm{x}_{\rm R})-\mu(\bm{x}))^2 = -\frac{A^2}{N^2}N_{\rm L}^2 + \frac{A}{N^2}(NA + 2N_{\rm M}\mu(\bm{x}_{\rm M}) - 2N_{\rm M}A)N_{\rm L}  +  \frac{1}{N}(N_{\rm M}\text{var}(\bm{x}_{\rm M}) + N_{\rm M}\mu(\bm{x}_{\rm M})^2 - N_{\rm M}A^2)  -\frac{1}{N^2}(N_{\rm M}\mu(\bm{x}_{\rm M}) + NA - N_{\rm M}A)^2 +A^2$.
	It can be shown that $\text{var}(N_{\rm L}|\bm{x}_{\rm M})$ is a concave quadratic function of $N_{\rm L}$. Hence, the maximum value of $\text{var}(N_{\rm L}|\bm{x}_{\rm M})$ is attained when $N_{\rm L} = \left\langle \frac{N}{2} + \frac{N_{\rm M}\mu(\bm{x}_{\rm M})}{A} - N_{\rm M} \right\rangle$.
	Since $x_{N_{\rm L}+1} = \Delta$ and $x_{N-N_{\rm R}} =x_{N_{\rm L}+1} +(N_{\rm M}-1)\Delta$, we have $\mu(\bm{x}_{\rm M}) = \frac{1}{2}(x_{N-N_{\rm R}} + x_{N_{\rm L}+1}) 
	= \frac{N_{\rm M}+1}{2}\Delta$.
	Then, we obtain $N_{\rm L} = \left\langle \frac{N}{2} + \frac{N_{\rm M}\left\lceil\frac{A}{\Delta}\right\rceil\Delta}{2A} - N_{\rm M} \right\rangle$.
	On one hand, since $\left\lceil\frac{A}{\Delta}\right\rceil \geq \frac{A}{\Delta}$, it follows that $\frac{N}{2} + \frac{N_{\rm M}\left\lceil\frac{A}{\Delta}\right\rceil\Delta}{2A} - N_{\rm M} 
	\geq \frac{N-N_{\rm M}}{2}$.
	On the other hand, since $\left\lceil\frac{A}{\Delta}\right\rceil < \frac{A}{\Delta}+1$, we have $\frac{N}{2} + \frac{N_{\rm M}\left\lceil\frac{A}{\Delta}\right\rceil\Delta}{2A} - N_{\rm M} 
	< \frac{N-N_{\rm M}}{2} + \frac{1}{2}$.
	Then, we obtain $N_{\rm L} = \left\lceil \frac{N-N_{\rm M}}{2} \right\rceil$.
	Accordingly, $N_{\rm R}$ is given by $N_{\rm R} = N-N_{\rm L}-N_{\rm M} 
	= \left\lfloor \frac{N-N_{\rm M}}{2} \right\rfloor$.
	
	Therefore, for the case of $A < (N-1)\Delta$, the optimal AVV $\bm{v}^{\rm SC}$ and APV $\bm{x}^{\rm SC}$ for problem (P1) are given by \eqref{vRL_main} and \eqref{xRL3_main}, respectively. Moreover, $\text{var}\left(\bm{x}^{\rm SC}\right)$ is given by
	\begin{align}\label{crbfin1}
		&\text{var}\left(\bm{x}^{\rm SC}\right) 
		= \frac{\Delta^2}{6N}N_{\rm M}(N_{\rm M}+1)(2N_{\rm M}+1) 
		+ \frac{A^2}{N}\left\lfloor \frac{N-N_{\rm M}}{2} \right\rfloor \notag\\
		&\quad - \frac{1}{N^2}\left( \frac{\Delta}{2}N_{\rm M}(N_{\rm M}+1) 
		+ \left\lfloor \frac{N-N_{\rm M}}{2} \right\rfloor A \right)^2.
	\end{align}
	Thus, the corresponding minimum CRB is
	\begin{align}\label{crbfin2}
		{\rm{CRB}}_u(\bm{x}^{\rm SC}) 
		= \frac{\sigma^2\lambda^2}{8\pi^2P|\beta|^2} 
		\frac{1}{N\,\text{var}\left(\bm{x}^{\rm SC}\right)}.
	\end{align}
	It is observed that ${\rm{CRB}}_u(\bm{x}^{\rm SC})$ decreases with $N$ in the order of $\mathcal{O}(N^{-1})$, while it decreases with $A$ in the order of $\mathcal{O}(A^{-2})$. This completes the proof of Theorem~2.

	\bibliographystyle{IEEEtran}
	\bibliography{IEEEabrv,IEEEexample}

\begin{thebibliography}{10}
\providecommand{\url}[1]{#1}
\csname url@samestyle\endcsname
\providecommand{\newblock}{\relax}
\providecommand{\bibinfo}[2]{#2}
\providecommand{\BIBentrySTDinterwordspacing}{\spaceskip=0pt\relax}
\providecommand{\BIBentryALTinterwordstretchfactor}{4}
\providecommand{\BIBentryALTinterwordspacing}{\spaceskip=\fontdimen2\font plus
\BIBentryALTinterwordstretchfactor\fontdimen3\font minus
  \fontdimen4\font\relax}
\providecommand{\BIBforeignlanguage}[2]{{%
\expandafter\ifx\csname l@#1\endcsname\relax
\typeout{** WARNING: IEEEtran.bst: No hyphenation pattern has been}%
\typeout{** loaded for the language `#1'. Using the pattern for}%
\typeout{** the default language instead.}%
\else
\language=\csname l@#1\endcsname
\fi
#2}}
\providecommand{\BIBdecl}{\relax}
\BIBdecl

\bibitem{jiang2021the}
W.~Jiang, B.~Han, M.~A. Habibi, and H.~D. Schotten, ``{The road towards 6G: A
  comprehensive survey},'' \emph{{IEEE} Open J. Commun. Soc.}, vol.~2, pp.
  334--366, Feb. 2021.

\bibitem{mailloux2005phased}
R.~J. Mailloux, \emph{{Phased Array Antenna Handbook}}.\hskip 1em plus 0.5em
  minus 0.4em\relax 2nd ed. Norwood, MA, USA: Artech House, 2005.

\bibitem{roberts2011sparse}
W.~Roberts, L.~Xu, J.~Li, and P.~Stoica, ``{Sparse antenna array design for
  MIMO active sensing applications},'' \emph{{IEEE} Trans. Antennas Propagat.},
  vol.~59, no.~3, pp. 846--858, Mar. 2011.

\bibitem{zhu2023MAMag}
L.~Zhu, W.~Ma, and R.~Zhang, ``Movable antennas for wireless communication:
  Opportunities and challenges,'' \emph{IEEE Commun. Mag.}, vol.~62, no.~6, pp.
  114--120, Jun. 2024.

\bibitem{zhu2025tutorial}
L.~Zhu, W.~Ma, W.~Mei, Y.~Zeng, Q.~Wu, B.~Ning, Z.~Xiao, X.~Shao, J.~Zhang, and
  R.~Zhang, ``A tutorial on movable antennas for wireless networks,''
  \emph{{IEEE} Commun. Surveys Tuts.}, 2025, early access, DOI:
  10.1109/COMST.2025.3546373.

\bibitem{zhao2009single}
S.~Zhao, H.~Yang, and H.~Yang, ``Single antenna spatial diversity,'' in
  \emph{Proc. IEEE Int. Conf. Wireless Commun., Netw., Mobile Comput. (WiCOM)},
  Beijing, China, Sep. 2009, pp. 1--4.

\bibitem{zhu2022MAmodel}
{L. Zhu, W. Ma, and R. Zhang}, ``Modeling and performance analysis for movable
  antenna enabled wireless communications,'' \emph{IEEE Trans. Wireless
  Commun.}, vol.~23, no.~6, pp. 6234--6250, Jun. 2024, arXiv accessed on 11
  Oct. 2022.

\bibitem{mei2024movable}
W.~Mei, X.~Wei, B.~Ning, Z.~Chen, and R.~Zhang, ``Movable-antenna position
  optimization: A graph-based approach,'' \emph{IEEE Wireless Commun. Lett.},
  vol.~13, no.~7, pp. 1853--1857, Jul. 2024.

\bibitem{ning2024movable}
B.~Ning, S.~Yang, Y.~Wu, P.~Wang, W.~Mei, C.~Yuen, and E.~Bj{\"o}rnson,
  ``Movable antenna-enhanced wireless communications: General architectures and
  implementation methods,'' \emph{IEEE Wireless Commun.}, 2025, early access,
  DOI: 10.1109/MWC.013.2400238.

\bibitem{tang2024secure}
J.~Tang, C.~Pan, Y.~Zhang, H.~Ren, and K.~Wang, ``Secure {MIMO} communication
  relying on movable antennas,'' \emph{IEEE Trans. Commun.}, vol.~73, no.~4,
  pp. 2159--2175, Apr. 2025.

\bibitem{zhu2023MAmultiuser}
L.~Zhu, W.~Ma, B.~Ning, and R.~Zhang, ``Movable-antenna enhanced multiuser
  communication via antenna position optimization,'' \emph{IEEE Trans. Wireless
  Commun.}, vol.~23, no.~7, pp. 7214--7229, Jul. 2024.

\bibitem{wu2023movable}
Y.~Wu, D.~Xu, D.~W.~K. Ng, W.~Gerstacker, and R.~Schober, ``Movable
  antenna-enhanced multiuser communication: Optimal discrete antenna
  positioning and beamforming,'' in \emph{Proc. IEEE Global Commun. Conf.
  (Globecom)}, Kuala Lumpur, Malaysia, Dec. 2023, pp. 7508--7513.

\bibitem{qin2024antenna}
H.~Qin, W.~Chen, Z.~Li, Q.~Wu, N.~Cheng, and F.~Chen, ``Antenna positioning and
  beamforming design for fluid antenna-assisted multi-user downlink
  communications,'' \emph{IEEE Wireless Commun. Lett.}, vol.~13, no.~4, pp.
  1073--1077, Apr. 2024.

\bibitem{cheng2023sum}
Z.~Cheng, N.~Li, J.~Zhu, X.~She, C.~Ouyang, and P.~Chen, ``Sum-rate
  maximization for movable antenna enabled multiuser communications,''
  \emph{arXiv preprint arXiv:2309.11135}, 2023.

\bibitem{yang2024flexible}
S.~Yang, J.~An, Y.~Xiu, W.~Lyu, B.~Ning, Z.~Zhang, M.~Debbah, and C.~Yuen,
  ``Flexible antenna arrays for wireless communications: Modeling and
  performance evaluation,'' \emph{arXiv preprint arXiv:2407.04944}, 2024.

\bibitem{hu2024power}
G.~Hu, Q.~Wu, K.~Xu, J.~Ouyang, J.~Si, Y.~Cai, and N.~Al-Dhahir, ``Fluid
  antennas-enabled multiuser uplink: A low-complexity gradient descent for
  total transmit power minimization,'' \emph{IEEE Commun. Lett.}, vol.~28,
  no.~3, pp. 602--606, Mar. 2025.

\bibitem{ma2022MAmimo}
W.~Ma, L.~Zhu, and R.~Zhang, ``{MIMO} capacity characterization for movable
  antenna systems,'' \emph{IEEE Trans. Wireless Commun.}, vol.~23, no.~4, pp.
  3392--3407, Apr. 2024, arXiv accessed on 11 Oct. 2022.

\bibitem{chen2023joint}
X.~Chen, B.~Feng, Y.~Wu, D.~W.~K. Ng, and R.~Schober, ``Joint beamforming and
  antenna movement design for moveable antenna systems based on statistical
  {CSI},'' in \emph{Proc. IEEE Global Commun. Conf. (Globecom)}, Kuala Lumpur,
  Malaysia, Dec. 2023, pp. 4387--4392.

\bibitem{yeyuqi2023fluid}
Y.~Ye, L.~You, J.~Wang, H.~Xu, K.-K. Wong, and X.~Gao, ``{Fluid
  antenna-assisted MIMO transmission exploiting statistical CSI},'' \emph{IEEE
  Commun. Lett.}, vol.~28, no.~1, pp. 223--227, Jan. 2024.

\bibitem{ma2023MAestimation}
W.~Ma, L.~Zhu, and R.~Zhang, ``Compressed sensing based channel estimation for
  movable antenna communications,'' \emph{IEEE Commun. Lett.}, vol.~27, no.~10,
  pp. 2747--2751, Oct. 2023.

\bibitem{xiao2023channel}
Z.~Xiao, S.~Cao, L.~Zhu, Y.~Liu, B.~Ning, X.-G. Xia, and R.~Zhang, ``Channel
  estimation for movable antenna communication systems: A framework based on
  compressed sensing,'' \emph{IEEE Trans. Wireless Commun.}, vol.~23, no.~9,
  pp. 11\,814--11\,830, Sep. 2024.

\bibitem{zhu2023MAarray}
L.~Zhu, W.~Ma, and R.~Zhang, ``Movable-antenna array enhanced beamforming:
  Achieving full array gain with null steering,'' \emph{IEEE Commun. Lett.},
  vol.~27, no.~12, pp. 3340--3344, Dec. 2023.

\bibitem{ma2024multi}
W.~Ma, L.~Zhu, and R.~Zhang, ``Multi-beam forming with movable-antenna array,''
  \emph{IEEE Commun. Lett.}, vol.~28, no.~3, pp. 697--701, Mar. 2024.

\bibitem{ZhuLP_satellite_MA}
L.~Zhu, X.~Pi, W.~Ma, Z.~Xiao, and R.~Zhang, ``Dynamic beam coverage for
  satellite communications aided by movable-antenna array,'' \emph{IEEE Trans.
  Wireless Commun.}, vol.~24, no.~3, pp. 1916--1933, Mar. 2025.

\bibitem{zhu2024nearfield}
{L. Zhu, W. Ma, Z. Xiao, and R. Zhang}, ``Movable antenna enabled near-field
  communications: Channel modeling and performance optimization,'' \emph{IEEE
  Trans. Commun.}, 2025, early access, DOI: 10.1109/TCOMM.2025.3547783.

\bibitem{shao2024discrete}
X.~Shao, R.~Zhang, Q.~Jiang, and R.~Schober, ``{6D} movable antenna enhanced
  wireless network via discrete position and rotation optimization,''
  \emph{IEEE J. Select. Areas Commun.}, vol.~43, no.~3, pp. 674--687, Mar.
  2025.

\bibitem{shao2024Mag6DMA}
X.~Shao and R.~Zhang, ``{6DMA} enhanced wireless network with flexible antenna
  position and rotation: Opportunities and challenges,'' \emph{IEEE Commun.
  Mag.}, vol.~63, no.~4, pp. 121--128, Apr. 2025.

\bibitem{shao2024exploiting}
X.~Shao, R.~Zhang, and R.~Schober, ``Exploiting six-dimensional movable antenna
  for wireless sensing,'' \emph{IEEE Wireless Commun. Lett.}, vol.~14, no.~2,
  pp. 265--269, Feb. 2025.

\bibitem{zhuravlev2015experi}
A.~Zhuravlev, V.~Razevig, S.~Ivashov, A.~Bugaev, and M.~Chizh, ``Experimental
  simulation of multi-static radar with a pair of separated movable antennas,''
  in \emph{Proc. IEEE International Conf. Microwaves Commun. Antennas Electron.
  Syst. (COMCAS)}, Nov. 2015, pp. 1--5.

\bibitem{hinske2008using}
S.~Hinske and M.~Langheinrich, ``Using a movable {RFID} antenna to
  automatically determine the position and orientation of objects on a
  tabletop,'' in \emph{Smart Sensing Context, 3rd Eur. Conf.}\hskip 1em plus
  0.5em minus 0.4em\relax Springer, 2008, pp. 14--26.

\bibitem{ma2024MAsensing}
W.~Ma, L.~Zhu, and R.~Zhang, ``Movable antenna enhanced wireless sensing via
  antenna position optimization,'' \emph{IEEE Trans. Wireless Commun.},
  vol.~23, no.~11, pp. 16\,575--16\,589, Nov. 2024.

\bibitem{chen2025MAISACopt}
L.~Chen, M.-M. Zhao, M.-J. Zhao, and R.~Zhang, ``Antenna position and
  beamforming optimization for movable antenna enabled {ISAC}: Optimal
  solutions and efficient algorithms,'' \emph{arXiv preprint arXiv:2502.14198},
  2025.

\bibitem{wang2025MAnearsensing}
Y.~Wang, W.~Mei, X.~Wei, B.~Ning, and Z.~Chen, ``Antenna position optimization
  for movable antenna-empowered near-field sensing,'' \emph{arXiv preprint
  arXiv:2502.03169}, 2025.

\bibitem{ma2025MAISAC}
W.~Ma, L.~Zhu, and R.~Zhang, ``Movable antenna enhanced integrated sensing and
  communication via antenna position optimization,'' \emph{arXiv preprint
  arXiv:2501.07318}, 2025.

\bibitem{shao2022target}
X.~Shao, C.~You, W.~Ma, X.~Chen, and R.~Zhang, ``Target sensing with
  intelligent reflecting surface: Architecture and performance,'' \emph{IEEE J.
  Select. Areas Commun.}, vol.~40, no.~7, pp. 2070--2084, Jul. 2022.

\bibitem{grantcvx}
M.~Grant and S.~Boyd, \emph{{(Mar. 2014). CVX: MATLAB Software for Disciplined
  Convex Programming, Version 2.1}}, [Online]. Available: http://cvxr.com/cvx.

\end{thebibliography}
	
\end{document}